\documentclass[reqno]{amsart}

\usepackage{amsfonts,amssymb}
\usepackage{graphicx}
\usepackage{color}
\usepackage[linktocpage=true,colorlinks=true, linkcolor=blue, citecolor=red, urlcolor=green]{hyperref}

\usepackage[paperwidth=595pt,paperheight=842pt,margin=5pt]{geometry}
\DeclareMathAlphabet{\mathbbm}{U}{bbm}{m}{n}

\DeclareSymbolFont{ltrs}     {OT1}{pzc}{m}{it}
\DeclareSymbolFont{ltrsa}     {OMS}{cmsy}{m}{n}
\DeclareSymbolFont{ltrsA}{U}{txmia}{m}{it}
\DeclareSymbolFont{symbolsC}{U}{txsyc}{m}{n}
\DeclareSymbolFont{ltrsB}{U}{rsfs}{m}{n}

\DeclareSymbolFontAlphabet{\mfrak}{ltrsA}
\DeclareMathAlphabet{\mathpzc}{OT1}{pzc}{m}{it}
\DeclareMathAlphabet{\mathrsfs}{U}{rsfs}{m}{n}

\def\beq{\begin{equation}}
\def\eeq{\end{equation}}
\def\beqn{\begin{equation*}}
\def\eeqn{\end{equation*}}
\def\beqa{\begin{eqnarray}}
\def\eeqa{\end{eqnarray}}
\def\beqs{\begin{eqnarray*}}
\def\eeqs{\end{eqnarray*}}

\def\nn{\nonumber}

\def\Z{\mathbb{Z}}
\def\R{\mathbb{R}}

\def\C{\mathbb{C}}

\def\Om{\Omega}
\def\spr{\cdot}
\def\di{\partial}
\def\llb{[\![}
\def\rrb{]\!]}
\def\vac{\boldsymbol{1}}

\def\e{\mathrm e}

\def\z{\mathrm z}
\def\w{\mathrm w}
\def\u{\mathrm u}
\def\vv{\mathrm v}

\def\T{\mathrm T}
\def\xo{x}
\def\ii{\mathrm i}

\def\cD{\mathfrak{c}_D}
\def\so{\mathfrak{so}}
\def\SO{\mathrm{SO}}
\def\sl{\mathfrak{sl}}
\def\ph{\varphi}
\def\de{\delta}
\def\io{\iota}
\def\harm{\mathrm{har}}

\DeclareMathOperator\Hom{Hom}
\DeclareMathOperator\End{End}

\DeclareMathOperator\ad{ad}

\newtheorem{theorem}{Theorem}
\newtheorem{lemma}[theorem]{Lemma}
\newtheorem{proposition}[theorem]{Proposition}
\newtheorem{corollary}[theorem]{Corollary}

\theoremstyle{definition}
\newtheorem{definition}{Definition}

\theoremstyle{remark}
\newtheorem{remark}{Remark}

\numberwithin{equation}{section}

\numberwithin{theorem}{section}
\numberwithin{definition}{section}
\numberwithin{remark}{section}

\newcommand\thref{Theorem \ref}
\newcommand\leref{Lemma \ref}
\newcommand\prref{Proposition \ref}

\newcommand\deref{Definition \ref}

\newcommand\reref{Remark \ref}
\newcommand\seref{Sect.\ \ref}

\newcounter{DEFINi}

\newcounter{STATEM}
\renewcommand{\theSTATEM}{\arabic{section}.\arabic{STATEM}}

\newcounter{EXAMp}
\renewcommand{\theEXAMp}{\arabic{section}.\arabic{EXAMp}}

\newcounter{Remk}
\renewcommand{\theRemk}{\arabic{section}.\arabic{Remk}}

\renewcommand{\theequation}{\arabic{section}.\arabic{equation}}

\def\ZP{\Z_{\geqslant 0}}

\begin{document}

\title[Reconstruction of vertex algebras in higher dimensions]{Reconstruction of vertex algebras in even higher dimensions}

\author{Bojko N. Bakalov}
\address{Department of Mathematics, North Carolina State University, Raleigh, NC 27695-8205, USA}
\email{bojko\_bakalov@ncsu.edu}

\author{Nikolay M. Nikolov}
\address{Institute for Nuclear Research and Nuclear Energy of Bulgarian Academy of Sciences, Tzarigradsko chaussee 72, BG-1784, Sofia, Bulgaria}
\address{Faculty of Mathematics and Informatics, Sofia University ''St. Kliment Ohridski'', James Bourchier 5, BG-1164, Sofia, Bulgaria}
\email{nikolov.qft@gmail.com}
\email{mitov@inrne.bas.bg}

\date{December 19, 2022; Revised June 25, 2023}



\setcounter{footnote}{0}

\begin{abstract}
Vertex algebras in higher dimensions correspond to models of quantum field theory with global conformal invariance.
Any vertex algebra in dimension $D$ admits a restriction to a vertex algebra in any lower dimension and, in particular, to dimension one.
In the case when $D$ is even, we find natural conditions under which the converse passage is possible.
These conditions include a unitary action of the conformal Lie algebra with a positive energy, which is given by local endomorphisms and obeys certain integrability properties.
\end{abstract}

\maketitle

\tableofcontents

\section{Introduction}\label{Se.1}

In Wightman's axiomatic approach to Quantum Field Theory (QFT), it is possible to restrict quantum fields to time-like submanifolds. This fact is a consequence of the most general and fundamental physical principles such as positivity of energy and locality. Thus, a natural question arises: which QFT models in low-dimensional space-time can be obtained in this way? In other words, which models in low dimensions can generate models in higher dimensions, and what additional structure is needed for this?

In this paper, we fully solve this problem within a special class of quantum field models, those with Global Conformal Invariance (GCI). These models were introduced in \cite{NT01}, where it was shown that they are characterized as conformally invariant models with rational correlation functions. This rationality allows the QFT models with GCI to be described in a purely algebraic way in terms of the vertex algebras in higher dimensions introduced in \cite{N05} and further developed in \cite{BN06}. One of the main results of \cite{N05} (see Theorems 9.2 and 9.3), is a precise one-to-one correspondence between conformal vertex algebras in higher dimensions with Hermitian structure and models of Wightman's axioms that possess GCI.

The present work is done entirely within the formalism of vertex algebras in higher dimensions, which are briefly reviewed in Sect.\ \ref{Se.2} below.
A vertex algebra in dimension $D$ admits a restriction to any lower dimension $D'<D$ (see Sect.\ \ref{sec2.3}). In particular, vertex algebras in dimension $1$ are the same as the (chiral) vertex algebras introduced by Borcherds in \cite{B86};  see Sect.\ \ref{Se.2.1-v3}.
They are well studied in connection with representation theory \cite{FLM88,K98,FB04,LL04} and $2$-dimensional conformal field theory \cite{BPZ84,G89,DMS97}.

Our main result is the Reconstruction Theorem \ref{TH}, which shows how one can construct a vertex algebra in even dimension $D$ from a chiral vertex algebra that possesses additional symmetry structure. In a different form, this result was previously announced in \cite{N05a,BN-ICM,BN08}. In this paper, we present a complete proof of the theorem.

Let us explain our construction in a heuristic way. Consider a vertex algebra $V$ in dimension $D$ with a state-field correspondence $Y$ (see Sect.\ \ref{Se.2}).
Suppose that $V$ is equipped with an action of the Lie algebra $\cD$ of infinitesimal conformal transformations of the complex Euclidean space $\C^D$, 
so that $Y$ is covariant under this action (see Sect.\ \ref{Se.2.3}). Additionally, we assume that
the action of $\cD$ on $V$ can be integrated to a group action $U(g)$ on $V$ for suitable elements $g$ of the conformal group (see \deref{main-def}; for general elements $g$ one needs a Hilbert space completion of $V$). 
Intuitively, for $a,b\in V$, we can think of $Y(a,\z)b$ as a conformally-covariant vector-valued function of $\z\in\C^D$. 
If we consider the restriction $\z \mapsto x\e_1$ to a line, where $\e_1 =(1,0,\dots,0) \in\C^D$, we obtain that $V$ is a chiral vertex algebra with
a state-field correspondence $Y_1(a,x)=Y(a,x\e_1)$. 

Conversely, if we are given a chiral vertex algebra $(V,Y_1)$, which is a restriction of a vertex algebra $(V,Y)$ in dimension $D$, then we can reconstruct $Y$ from $Y_1$ and the action of $\cD$.
Explicitly, $Y(a,\z)b$ is obtained from the adjoint action $U(g)Y_1(U(g)^{-1}a,x)U(g)^{-1}b$ of a conformal transformation $g$ that maps the pair of points $(0,x\e_1)$ to the pair $(0,\z)$;
see Eq.\ \eqref{vert4-nn1-mod090422} below.
In order for this construction to be well defined, we need a certain invariance condition for the action of the stabilizer of a pair $(0,x\e_1)$ in the conformal group action.  We have found two equivalent formulations of this invariance. The first consists of explicit commutation relations between the Lie algebra $\cD$ and the vertex operators $Y_1(a,x)$ of the initial chiral vertex algebra (see Proposition \ref{ploc2}). The second condition is more abstract and is given in terms of the notion of a local endomorphism of a vertex algebra, introduced in Sect.\ \ref{sloc}.
We further develop this notion, which may be of independent interest and is related to the pseudoderivations of Etingof and Kazhdan \cite{EK00} (see also \cite{L05}).  
In fact, local endomorphisms have already found applications in other works such as \cite{Ne15,Ne16,Ne16a}.
We also remark that there is an analogy between our algebraic construction and the geometric dimensional reduction introduced in \cite{N87} (see also \cite{NP03-1,NP03-2}).

After we reconstruct the fields $Y(a,\z)$ from their restrictions $Y(a,x\e_1)$, as outlined above, we have to show that the state-field correspondence $Y$ endows $V$ with the structure of a vertex algebra in dimension $D$. One of the most important axioms to check is the locality of the fields. Let us explain heuristically the idea behind its proof (see Sect.\ \ref{sse5.7n}).
Locality is a statement about the product $Y(a,\z)Y(b,\w)c$ for three mutually non-isotropic points $0$, $\z$, $\w$ in $\C^D$. It is known that the conformal group acts transitively on triples of mutually non-isotropic complex points. In particular, they can be mapped into a line, and in this way locality in dimension $D$ can be derived from locality in dimension $1$.

Another axiom of a vertex algebra requires that $Y(a,\z)b$ become a formal power series in $\z=(z^1,\dots,z^D)$ after multiplication by some power of the Euclidean square-length $\z^2=(z^1)^2+\dots+(z^D)^2$. However, our construction only defines $Y(a,\z)b$ as a doubly-infinite formal series in $\z$ and $1/\z^2$ (see \cite{N05,BN06} for more details about such series).
Proving the bounds of the poles in $\z^2$ turned out to be challenging and solving this problem is one of the main contributions of this paper.
This is where we essentially used unitarity (see Sect.\ \ref{sse5.6n}),
while the rest of our construction does not rely on any Hilbert space positivity.

It is remarkable that in low space-time dimensions ($1$ or $2$) there are large classes of non-trivial QFT models, while in dimension $4$ and higher only the free field models are mathematically fully developed within the established axiomatic frameworks of Wightman and Haag--Kastler (see \cite{BLOT90,H96,A99}).
With the present work, we suggest an intriguing bridge between lower- and higher-dimensional models. 
The long-term goal of constructing QFT models in higher dimensions is one of the main application prospects of our results.


\subsection*{Notation}

All vector superspaces in this paper are over the field $\C$ of complex numbers.
We denote by $\Z$ and $\ZP$ the sets of integers and non-negative integers, respectively. 
We will mainly follow the notation from our previous works \cite{N05,BN06} but for completeness we list the most important notation here.

Throughout the paper, we fix a positive integer $D$.
We will denote by $\z=(z^1,\dots,z^D)$, $\w=(w^1,\dots,w^D)$, etc.,
vectors of formal variables.
We let
$\z \spr \w$ $=$ $z^1w^1 + \cdots + z^Dw^D$, $\z^2 := \z \spr \z$, and
$\di_{z^{\alpha}} = \di / \di z^{\alpha}$ for $\alpha = 1,\dots,D$.
For a vector space $V$, we denote by $V\llb \z \rrb$
the space of formal power series in $z^1,\dots,z^D$ with coefficients in $V$.
The space $V \llb \z \rrb_{\z^2}$ is the space of series of the form
$(\z^2)^{-N} v(\z)$ where $v (\z) \in V \llb \z \rrb$ and $N \in \ZP$.
Similarly, we will use the spaces
$V \llb \z,\w \rrb_{\z^2 \w^2}$.

A complete introduction to the various spaces of formal series
that appear in the theory of vertex algebras in higher dimensions
can be found in \cite[Sect.\ 1]{N05} and \cite[Sect.\ 1]{BN06}.
In particular, 
we recall that 
$V \llb \z \rrb_{\z^2}$ is a module over $\C \llb \z \rrb_{\z^2}$,
as well as a differential module with derivations
$\di_{z^{\alpha}}$ ($\alpha = 1,\dots,D$).
In dimension $D=1$, we have $\z \equiv z^1$ and we shall just denote it by $z$.
In this case, the space $V \llb \z \rrb_{\z^2}$ reduces to the space of Laurent series $V(\!(z)\!)$ $:=$ $V\llb z\rrb [z^{-1}]$ $\equiv$ $V\llb z\rrb_z$ $\equiv$ $V\llb z\rrb_{z^2}$.

\section{Vertex Algebras in Dimension $1$ and $D$}\label{Se.3}
In this section, we first briefly recall the definition and main properties of vertex algebras in dimension $1$ (for more details, see \cite{FLM88,K98,FB04,LL04}). 
Then we review the theory of vertex algebras in higher dimensions, which were introduced in \cite{N05}
and developed further in \cite{BN06}.

\subsection{Vertex algebras in dimension $1$}\label{Se.2.1-v3}

The notion of a vertex algebra introduced by Borcherds \cite{B86}
provides a rigorous algebraic description of two-dimensional 
chiral conformal field theory (see e.g.\ \cite{BPZ84,G89,DMS97}).
A \textbf{vertex algebra} is a vector superspace $V$ 
(space of states) with a distinguished even vector $\vac\in V$ 
(vacuum vector), together with an even linear map 
(state-field correspondence):
\begin{equation}\label{vert2}
Y(\cdot,z)\cdot \colon V \otimes V \to V(\!(z)\!) = V\llb z\rrb [z^{-1}] \,.
\end{equation}
For every state $a\in V$, we have the \textbf{field}
$Y(a,z) \colon V \to V(\!(z)\!)$. This field can also be viewed as
a formal power series from $(\End V)\llb z,z^{-1} \rrb$, which after applied
to any vector involves only finitely many negative powers of $z$.
The coefficients in front of powers of $z$ in this expansion are known as the
\textbf{modes} of $a$:
\begin{equation}\label{vert4}
Y(a,z)b = \sum_{n=-\infty}^{N_{a,b}-1} a_{(n)}b \, z^{-n-1} 
\,, \qquad a, b \in V \,.
\end{equation}
We can think of $a_{(n)}b$ as an infinite sequence of products that vanish for sufficiently large $n\geqslant N_{a,b}$ (depending
on $a$ and $b$) and are subject to certain axioms.

First, the vacuum vector $\vac$ behaves as a partial identity:
\begin{equation}\label{vert6}
a_{(m)} \vac = \de_{m,-1} a \,, \quad
\vac_{(n)} a = \de_{n,-1} a \,, \qquad
m,n \in \Z \,, \; m\geqslant -1 \,.
\end{equation}
Next, define the \textbf{translation operator} $T \in \End V$ by
$Ta = a_{(-2)} \vac$. Notice that $T$ is even.
Then all fields $Y(a,z)$ are required to be
\textbf{translation covariant}:
\begin{equation}\label{vert7}
[T,Y(a,z)] = \di_z Y(a,z) \,.
\end{equation}
Finally, all fields $Y(a,z)$ in a vertex algebra must be
\textbf{local} with each other, which means that for all
$a,b \in V$ there is a non-negative integer $N_{a,b}$
such that
\begin{equation}\label{vert8}
(z-w)^{N_{a,b}} Y(a,z) Y(b,w) = (z-w)^{N_{a,b}} (-1)^{p_a p_b} Y(b,w) Y(a,z) \,,
\end{equation}
where $p_a$ denotes the parity of $a$.
This completes the definition of a vertex algebra.
We will denote a vertex algebra as $(V,Y,\vac,T)$, or simply as $V$.
Note that our notion of a vertex algebra is a superspace, but for brevity
we call it a vertex algebra instead of a vertex superalgebra, following the convention of \cite{K98}.

We can take the same $N_{a,b}$ in Eqs.\ \eqref{vert4} and \eqref{vert8}.
It follows from the vacuum and translation axioms 
\eqref{vert6}, \eqref{vert7} that
\begin{equation}\label{vert9}
Y(a,z)\vac = e^{zT} a \,, \qquad Y(\vac,z)a = a \,.
\end{equation}
The locality condition \eqref{vert8} implies the
Borcherds \textbf{commutator formula}
$(a,b\in V, \, m,n\in\Z)$:
\begin{equation}\label{vert10}
\begin{split}
[a_{(m)}, b_{(n)}] :=&\; a_{(m)} b_{(n)} - (-1)^{p_a p_b} b_{(n)} a_{(m)} \\
=&\, \sum_{j=0}^{N_{a,b}-1} \binom{m}{j} (a_{(j)}b)_{(m+n-j)} \,,
\end{split}
\end{equation}
which can be rewritten equivalently as
\begin{equation}\label{vert12}
[a_{(m)}, Y(b,w)]
= \sum_{j=0}^{N_{a,b}-1} \binom{m}{j} w^{m-j} \, Y(a_{(j)}b,w) \,.
\end{equation}

Among other important results in the theory of vertex algebras, let us mention Dong's Lemma (\cite[Lemma 3.2]{K98}),
the Goddard Uniqueness Theorem (\cite[Theorem 4.4]{K98}),
and the Kac Existence Theorem (\cite[Theorem 4.5]{K98}).
The last theorem allows one to generate a vertex algebra from a collection of translation covariant local fields.

\subsection{The definition of vertex algebras in dimension $D$}\label{Se.2}

From now on, we will fix a positive integer $D$ and use the notation from the introduction; in particular,
$\z=(z^1,\dots,z^D)$ will be a vector of formal variables.
Let us recall from \cite[Definition 2.1]{N05}, \cite[Definition 4.1]{BN06}
the notion of a \textbf{vertex algebra} 
in dimension $D$. This is a vector superspace $V$, endowed with an even linear map
\begin{equation*}
V \otimes V \to V \llb \z \rrb_{\z^2} \,, \quad
a \otimes b \mapsto Y (a,\z) b \,, 
\end{equation*}
called a \textbf{state-field correspondence},
a system of commuting even endomorphisms $T_1,\dots,T_D$ of $V$
called \textbf{translation operators}, 
and an even vector $\vac \in V$ called the \textbf{vacuum},
subject to the following axioms.

\begin{enumerate}
\medskip
\item[$(a)$] \textit{Locality}:
The formal series
$f_{a,b,c} (\z,\w) := \bigl((\z-\w)^2\bigr)^{N_{a,b}} Y (a,\z) Y (b,\w) c$
takes values in $V \llb \z,\w\rrb_{\z^2\w^2}$ for some non-negative integer $N_{a,b} = N_{b,a}$,
and $f_{a,b,c} (\z,\w)$ $= (-1)^{p_a p_b} f_{b,a,c} (\w,\z)$
for all $a,b,c \in V$, where $p_a$ and $p_b$ are the parities of $a$ and $b$,
respectively.

\medskip
\item[$(b)$] \textit{Translation covariance}:
$[T_{\alpha}, Y(a,\z)] = \di_{z^{\alpha}} Y (a,\z)$ for all
$a \in V$, $\alpha =1,\dots,D$.

\medskip
\item[$(c)$] \textit{Vacuum}:
$T_{\alpha} \vac = 0$, $Y(\vac,\z)a = a$ and $Y(a,\z)\vac \in V\llb \z \rrb$ with
the property that
$Y(a,\z)\vac |_{\z = 0} = a$,
for all $a \in V$ and $\alpha = 1,\dots,D$.
\end{enumerate}

\medskip
We will denote a vertex algebra as $(V,Y,\vac,\T)$ or simply as $V$, where $\T$ is the collection $(T_1,$ $\dots,$ $T_D)$.
In the case $D=1$, the above definition reduces to that of the well-known chiral vertex algebras; see \cite{B86,FLM88,K98,FB04,LL04} and 
\seref{Se.2.1-v3}.

\begin{remark}\label{remdim}
In the above definition of a vertex algebra $V$ in dimension $D$, the word ``dimension" refers to the dimension of the vector
$\z=(z^1,\dots,z^D)$, which is a vector formal variable (i.e., $z^1,\dots,z^D$ are formal variables as is customary in the theory of
vertex algebras; see e.g., \cite{K98}). For some purposes, it is convenient to think of $\z$ as a vector in $\C^D$. 
From that point of view, $D$ denotes a \emph{complex} dimension.
However, it should not be confused with the dimension $\dim_\C V$ of $V$ as a complex vector superspace, which is almost always infinite
except in trivial cases. In quantum field theory (QFT), the vector variable $\z$ represents a point in the complex-analytic continuation of the complexified Minkowski space-time. This continuation exists as a consequence of the most general axiomatic principles of QFT (see for more details \cite[Chapt.\ 1]{K98} or \cite[Sect.\ 9]{N05a}).
\end{remark}

It is convenient to expand the series $Y(a,\z)b$ in a basis similarly to Eq.~(\ref{vert4}). 
To this end, we use bases of harmonic polynomials $\{h_{m,\sigma}\}_{\sigma = 1}^{\mathfrak{h}_m^D}$ $\subset$ $\C_m^{\harm} [\z]$ for each degree $m$ of homogeneity (cf.\ \cite[Sect.\ 1]{N05} or \cite[Sect.\ II.B]{BN06}):
$$
(\di_{z^1}^2 + \cdots + \di_{z^D}^2) \, h_{m,\sigma} (\z) \,=\, 0
\,,\quad
(z^1\di_{z^1} + \cdots + z^D\di_{z^D}) \, h_{m,\sigma} (\z) \,=\, m \, h_{m,\sigma} (\z)
\,.
$$
Then we have
\beq\label{vert4-nn}
Y (a,\z) b \, = \, 
\sum_{n \,=\, -N_{a,b}}^{\infty} \
\mathop{\sum}\limits_{m = 0}^{\infty} \
\mathop{\sum}\limits_{\sigma = 1}^{\mathfrak{h}_m^D} \
a_{\{n,m,\sigma\}}b \ (\z^2)^n \,h_{m,\sigma}(\z)
\,
\eeq
for some elements $a_{\{n,m,\sigma\}}b \in V$
(cf.\ \cite[Eq.\ (2.1)]{N05}, \cite[Eq.\ (2.9)]{BN06}).

\begin{remark}\label{remharm}
The expansion \eqref{vert4-nn} is an extension of the well-known expansion of an arbitrary homogeneous polynomial $f(\z) \in\C[\z]$ of degree $m$ in terms of harmonic polynomials:
\beq\label{harexp}
f(\z) = \sum_{0\leqslant n \leqslant \frac{m}2}
\mathop{\sum}\limits_{\sigma = 1}^{\mathfrak{h}_{m-2n}^D}
f_{n,m-2n,\sigma} \, (\z^2)^n \,h_{m-2n,\sigma}(\z)
\,,
\eeq
for unique coefficients $f_{n,m-2n,\sigma} \in\C$ (see e.g.\ \cite[Lemma 1.1]{N05}).
\end{remark}

Many important results from the theory of vertex algebras in dimension $1$ have generalizations to higher dimensions.
In particular, we have
\begin{equation}\label{vert9-hd}
Y(a,\z)\vac = e^{\z\spr \T} a \,, \qquad Y(\vac,\z)a = a \,,
\end{equation}
where $\z\spr \T$ stands for $\mathop{\sum}\limits_{\alpha=1}^D z^{\alpha} T_{\alpha}$ (cf.\ \cite[Proposition 3.2]{N05}).
The analog of Goddard's Uniqueness Theorem is Theorem 3.1 of \cite{N05}, which is the vertex algebra counterpart of the Reeh--Schlieder Theorem in 
Wightman's approach to quantum field theory.
There are also analogs of Dong's Lemma and the Kac Existence Theorem, which are Proposition 3.4 and Theorem 4.1 of \cite{N05}, respectively.

\subsection{Restriction to lower dimensions}\label{sec2.3}

Now let us assume that $D\geqslant2$ and fix a positive integer $D' < D$.
We will denote vectors of dimension $D'$ by $\z'$, $\w'$, etc.
Notice that
\begin{equation*}
(\z')^2 = \z^2 \,, \; (\z'-\w')^2 = (\z-\w)^2
\quad\text{for}\quad \z = (\z',0,\dots,0) \,, \;
\w = (\w',0,\dots,0) \,.
\end{equation*}
This induces a restriction morphism 
$V \llb \z \rrb_{\z^2} \to V \llb \z' \rrb_{\z'{}^2}$,
which agrees with the module actions of 
$\C \llb \z \rrb_{\z^2}$ and $\C \llb \z' \rrb_{\z'{}^2}$
under the algebra homomorphism
$\C \llb \z \rrb_{\z^2} \to \C \llb \z' \rrb_{\z'{}^2}$ induced also by the restriction.
Moreover, the operations of taking a partial derivative $\di_{z^\alpha}$ $(\alpha=1,\dots,D')$ and evaluation $\z = (\z',0,\dots,0)$ commute.
Hence, the restriction of the state-field correspondence
\beq\label{Eq.2.1}
Y_{D'} (a,\z') b := Y (a,\z) b \big|_{\z = (\z',0,\dots,0)}
\eeq
makes sense, and
$V$ endowed with
the restricted state-field correspondence (\ref{Eq.2.1}), the translation endomorphisms
$T_1,\dots,T_{D'}$ and the same vacuum $\vac \in V$ is a vertex algebra in dimension $D'$.
It is called a $D'$-dimensional \textbf{restriction} of $V$.

\subsection{Conformal vertex algebras}\label{Se.2.3}

The complex Lie algebra $\cD$ of \textbf{infinitesimal conformal transformations} of 
$\C^D$ is spanned by the infinitesimal translations $T_{\alpha}$,
dilation $H$, rotations $\Omega_{\alpha\beta} = -\Omega_{\beta\alpha}$,
and special conformal transformations $C_{\alpha}$. These generators
satisfy the following relations:
\begin{equation}\label{cDcomm}
\begin{split}
[H,\Omega_{\alpha\beta}] &=[T_{\alpha},T_{\beta}]=[C_{\alpha},C_{\beta}]=0 \,, \\
[\Omega_{\alpha\beta},T_{\gamma}] &=\delta_{\alpha\gamma} T_{\beta}-\delta_{\beta\gamma} T_{\alpha} \,, \qquad\quad\;\;\;
[\Omega_{\alpha\beta},C_{\gamma}] =\delta_{\alpha\gamma} C_{\beta}-\delta_{\beta\gamma} C_{\alpha} \,, \\
[H,T_{\alpha}] &=T_{\alpha}, \quad\; [H,C_{\alpha}]=-C_{\alpha} \,, \quad\;\;
[T_{\alpha},C_{\beta}]=2\delta_{\alpha\beta} H - 2\Omega_{\alpha\beta} \,, \\
[\Omega_{\alpha_1\beta_1},\Omega_{\alpha_2\beta_2}]
&=\delta_{\alpha_1\alpha_2}\Omega_{\beta_1\beta_2}
+
\delta_{\beta_1\beta_2}\Omega_{\alpha_1\alpha_2}
- \delta_{\alpha_1\beta_2}\Omega_{\beta_1\alpha_2}-\delta_{\beta_1\alpha_2}\Omega_{\alpha_1\beta_2} \,. 
\end{split}
\end{equation}
Note that we have an isomorphism of Lie algebras $\cD\cong\so(D+2,\C)$,
and the infinitesimal rotations $\Omega_{\alpha\beta}$ span a subalgebra isomorphic to $\so(D,\C)$.

A \textbf{conformal vertex algebra} \cite[Definition 7.1]{N05} is a vertex algebra whose underlying vector superspace is a $\mathfrak{c}_D$--module,
which extends the action of the infinitesimal translations $T_{\alpha}$ and satisfies the following properties.

\begin{enumerate}
\medskip
\item[$(a)$] \textit{Integrability}: Every vector
is contained in a finite-di\-men\-si\-o\-nal subspace invariant under all $\Om_{\alpha\beta}$.

\medskip
\item[$(b)$] \textit{Energy positivity}:
$H$ is diagonalizable with eigenvalues in $\frac12\ZP$.

\medskip
\item[$(c)$] \textit{Conformal equivariance}:
\begin{align}\label{CVA1}
\bigl[H,Y(a,\z)\bigr] &= Y(Ha,\z) + \z\cdot\di_\z Y(a,\z),
\\ \label{CVA2}
\bigl[\Om_{\alpha\beta},Y(a,\z)\bigr] &= Y(\Om_{\alpha\beta}a,\z)
+ (z^\alpha \di_{z^\beta} - z^\beta \di_{z^\alpha}) Y(a,\z),
\\ \label{CVA3}
\bigl[C_\alpha,Y(a,\z)\bigr] &= Y(C_\alpha a,\z) - 2 z^\alpha Y(Ha,\z)
- 2 \sum_{\beta=1}^D z^\beta Y(\Om_{\alpha\beta} a,\z)
 \\
&+ (\z^2 \di_{z^\alpha} - 2 z^\alpha \z \cdot\di_\z) Y(a,\z). \nn
\end{align}

\medskip
\item[$(d)$] \textit{Vacuum}: $\mathfrak{c}_D \vac = 0$.
\end{enumerate}

\medskip

When we want to specify all the structure of a conformal vertex algebra in dimension $D$, we will denote it as
$(V,Y,\vac,\mathfrak{c}_D)$.

\begin{remark}
If we restrict the above definition to the case $D=1$, then
$\mathfrak{c}_1 \cong \so(3,\C) \cong \sl(2,\C)$ and our notion
of a conformal vertex algebra coincides with that of a
{M\"obius} vertex algebra from \cite{K98}.
\end{remark}

\begin{remark}
It follows from \thref{tloc1} below that the above condition $(d)$ is a consequence of condition $(c)$.
\end{remark}

We introduce a real form of $\mathfrak{c}_D$ by the following anti-linear anti-involution:
\beq\label{eq2.1}
H^* \, = \, H,\quad \Omega_{\alpha\beta}^* \, = \, -\Omega_{\alpha\beta},\quad T_{\alpha}^* \, = \, C_{\alpha},\quad C_{\alpha}^* \, = \, T_{\alpha} \,,
\eeq
so that $(\lambda A)^* = \bar\lambda A^*$, $A^{**} = A$, and $[A,B]^* = [B^*,A^*]$ for all $\lambda\in\C$ and $A,B\in\cD$.
A representation of $\cD$ on a complex vector space $V$ is called \textbf{unitary} if $V$ is equipped with a positive-definite Hermitian product such that
the above anti-involution coincides with the Hermitian conjugation.

\section{Local Endomorphisms of Vertex Algebras}\label{sloc}

In this section, we introduce the notion of a local endomorphism of a vertex algebra (in any dimension)
and present a formula for its commutator with the fields of the vertex algebra. In the case of dimension $1$,
local endomorphisms are related to the pseudoderivations of \cite{EK00}.
We consider the example of a local action of $\cD$, which will be important for the rest of the paper.

\subsection{Local endomorphisms}\label{locend}

\begin{definition}\label{dloc1}
Let $(V,Y,\vac,\T)$ be a vertex algebra in dimension $D$ (including the case $D=1$).
We say that a linear operator $X\in\End V$ is \textbf{local} if\/ $[X,Y(a,\z)]$ is local with respect to $Y(b,\z)$ for all $a,b \in V$. 
\end{definition}

As before, brackets here denote the \emph{supercommutator}; for example,
\begin{equation}\label{supercom}
[X,Y(a,\z)] := X \, Y(a,\z) - (-1)^{p_X p_a} Y(a,\z) \, X \,,
\end{equation}
where $p_X$ and $p_a$ are the parities of $X$ and $a$, respectively.
It is easy to see that $[X,Y(a,\z)]$
is again a field for any $X\in\End V$, i.e., $[X,Y(a,\z)]v \in V\llb \z \rrb_{\z^2}$ for every $v\in V$. 
Explicitly, the condition that $X$ is \emph{local} means that, for all $a,b \in V$, there is a non-negative integer $N_{a,b}^X$ such that
\begin{equation}\label{loc1}
\bigl((\z-\w)^2\bigr)^{N_{a,b}^X} \bigl[[X,Y(a,\z)], Y(b,\w)\bigr] = 0 \,.
\end{equation}

\begin{remark}
\begin{enumerate}
\item[$(a)$]
In general, the space of all local endomorphism of\/ $V$ is not a Lie superalgebra under the supercommutator. 
However, under some additional assumptions like the conditions of Theorem \ref{tloc1} below, one obtains a Lie superalgebra.

\medskip
\item[$(b)$]
In any vertex algebra, the translation operators $T_\alpha$ are local. More generally, any derivation of $V$ is a local endomorphism.

\medskip
\item[$(c)$]
Other examples of local endomorphisms in dimension $D=1$ are provided by all modes of fields: the commutator formula \eqref{vert12} implies that $a_{(m)} \in\End V$ 
are local for $a\in V$, $m\in\Z$.

\medskip
\item[$(d)$]
Suppose that the vertex algebra $V$ is generated by a collection of fields $\{\ph_a(\z)\}$.
It follows from Dong's Lemma that a linear operator $X\in\End V$ is local if and only if $[X,\ph_a(\z)]$ is local with respect to $\ph_b(\z)$ for all $a,b$.
\end{enumerate}
\end{remark}

\begin{theorem}\label{tloc1}
For a vertex algebra\/ $V$ and a linear operator\/ $X$ of\/ $V$, the following two conditions are equivalent$:$
\begin{enumerate}
\item[$(i)$]
$X$ is a local endomorphism of\/ $V$ such that\/
$X\vac=0$ and\/ $(e^{-\ad(\z\spr\T)} X)a \in V[\z]$ for every\/ $a\in V$.

\medskip
\item[$(ii)$]
There is a linear map\/ $X(\z)\colon V\to V[\z]$ such that\/ $X(0) = X$ and\/
$[X,Y(a,\z)] = Y(X(\z)a, \z)$ for all\/ $a\in V$.
\end{enumerate}
In either of these two cases, one has\/ $X(\z)=e^{-\ad(\z\spr\T)} X$ and
\begin{equation}\label{loc2}
[X,Y(a,\z)] = Y\bigl((e^{-\ad (\z\spr\T)} X)a, \z\bigr) \,.
\end{equation}
\end{theorem}
\begin{proof}
If $(i)$ holds, then $X(\z):=e^{-\ad(\z\spr\T)} X$ obviously satisfies $X(0) = X$ and $X(\z)a\in V[\z]$ for all $a\in V$ by assumption. Hence, $(ii)$ will follow from formula (\ref{loc2}).
To prove Eq.\ (\ref{loc2}), we observe that $[X,Y(a,\z)]$ is a field on $V$, which is local with respect to all fields $Y(b,\z)$ for $b\in V$.
Furthermore, by Eq.\ \eqref{vert9-hd}, we have
\[
[X,Y(a,\z)] \vac = X Y(a,\z)\vac = X e^{\z \spr \T} a = e^{\z\spr\T} \bigl((e^{-\ad (\z\spr\T)} X)a \bigr) \,.
\]
Hence, Eq.\ (\ref{loc2}) follows from the higher-di\-men\-sio\-nal analog of 
Goddard's Uniqueness Theorem \cite[Theorem 3.1]{N05}.

Conversely, suppose that $(ii)$ holds. Then $[X,Y(a,\z)] = Y(X(\z)a, \z)$ implies that $X$ is a local endomorphism of $V$.
Setting $a=\vac$ in this formula, we get
\[
0 = [X,id] = [X,Y(\vac,\z)] = Y(X(\z)\vac, \z) \,.
\]
Hence, applying the right-hand side to $\vac$ and setting $\z=0$, we obtain $X\vac = X(0)\vac = 0$.
To finish the proof, it remains to show that $(e^{-\ad(\z\spr\T)} X)a \in V[\z]$ for all $a\in V$.

Fix an index $\alpha$, apply the operator $\ad T_{\alpha}$ to both sides of formula
$[X,Y(a,\z)] = Y(X(\z)a, \z)$ and use the translation covariance $[T_{\alpha}, Y(a,\z)] = \di_{z^{\alpha}} Y (a,\z)$.
We obtain:
\[
\bigl[[T_{\alpha},X],Y(a,\z)\bigr] + [X,\di_{z^{\alpha}} Y(a,\z)] = \di_{w^{\alpha}} Y(X(\z)a, \w) \big|_{\w=\z} \,.
\]
The second summand above is equal to
\[
[X,\di_{z^{\alpha}} Y(a,\z)] = \di_{z^{\alpha}} [X,Y(a,\z)] = \di_{z^{\alpha}} \bigl(Y(X(\z)a, \z) \bigr) \,.
\]
Therefore,
\[
\bigl[[T_{\alpha},X],Y(a,\z)\bigr]  = -Y\bigl((\di_{z^{\alpha}} X(\z)a), \z\bigr) \,.
\]
Repeating this process, we see that
\beq\label{loc7}
\bigl[ (\ad T_{\alpha_1}) \cdots (\ad T_{\alpha_k}) X, Y(a,\z)\bigr]  = (-1)^k Y\bigl((\di_{z^{\alpha_1}} \cdots \di_{z^{\alpha_k}} X(\z)a), \z\bigr) 
\eeq
for any $\alpha_1,\dots,\alpha_k \in\{1,\dots,D\}$. 
Since $X(\z)a$ is a polynomial in $\z$, there exists a positive integer $N_a$ (depending on $a$) such that the right-hand side above is zero
for all $k\geqslant N_a$ and any choice of $\alpha_1,\dots,\alpha_k$. This implies that
\[
\bigl[ (\ad T_{\alpha_1}) \cdots (\ad T_{\alpha_k}) X, Y(a,\z)\bigr]  = 0 \,, \qquad k\geqslant N_a \,.
\]
Applying this identity to the vacuum vector $\vac$ and setting $\z=0$, we obtain
\[
\bigl( (\ad T_{\alpha_1}) \cdots (\ad T_{\alpha_k}) X\bigr) a = 0 \,, \qquad k\geqslant N_a \,,
\]
where we used that $X\vac=T_\alpha\vac=0$ and $Y(a,\z)\vac |_{\z = 0} = a$.
Hence, $(e^{-\ad(\z\spr\T)} X)a \in V[\z]$.
\end{proof}

As is customary in the theory of vertex algebras, let us denote by $\io_{\z,\u}$ the Taylor expansion
\begin{equation}\label{loc4}
\io_{\z,\u} f(\z+\u) = e^{\u \cdot \di_{\z}} f(\z) \,,
\end{equation}
for polynomials $f(\z)$ or more generally Laurent series in $\z$.
The next result is a consequence of Eq.\ \eqref{loc7} from the proof of Theorem \ref{tloc1}.

\begin{corollary}\label{ploc3}
With the notation of Theorem \ref{tloc1}, we have{\rm:}
\begin{align}\label{loc5}
[T_\alpha,X(\u)] &= -\di_{u^\alpha} X(\u) \,, \qquad \alpha=1,\dots,D \,,
\\ \label{loc6}
[X(\u),Y(a,\z)] &= \io_{\z,\u} Y(X(\z+\u)a, \z) \,.
\end{align}
\end{corollary}

In dimension $D=1$, this corollary means that the map $X(\u)$ is a {pseudoderivation} of $V$ in the sense of Etingof and Kazhdan \cite{EK00} (see also \cite{L05}).

\subsection{Local action of $\cD$}\label{Se.4.1}

In any conformal vertex algebra $V$ in dimension $D$, the conformal equivariance \eqref{CVA1}--\eqref{CVA3} implies that the action of $\cD$ on $V$ is a \textbf{local action}, i.e., any $X\in\cD$ acts as a local endomorphism of $V$.
In fact, we have:

\begin{proposition}\label{ploc1}
Let\/ $(V,Y,\vac,\T)$ be a vertex algebra in dimension\/ $D$, equipped with an action of\/ $\cD$ such that
each\/ $T_\alpha\in\cD$ acts as the translation operator\/ $T_\alpha$ of\/ $V$
$(\alpha=1,\dots,D)$.
Then the action of\/ $\cD$ is local and annihilates the vacuum\/ $(\mathfrak{c}_D \vac = 0)$
if and only if the conformal equivariance relations \eqref{CVA1}--\eqref{CVA3} hold.
\end{proposition}
\begin{proof}
This follows from \thref{tloc1} and the commutation relations \eqref{cDcomm}. Indeed, from \eqref{cDcomm}, we find:
\begin{align*}
e^{-\ad(\z\spr\T)} H &= H+\z\spr\T \,, \\
e^{-\ad(\z\spr\T)} \Om_{\alpha\beta} &= \Om_{\alpha\beta} + z^\alpha T_\beta - z^\beta T_\alpha \,, \\
e^{-\ad(\z\spr\T)} C_\alpha &= C_\alpha - 2 z^\alpha H - 2 \sum_{\beta=1}^D z^\beta \Om_{\alpha\beta} + \z^2 \, T_\alpha - 2 z^\alpha (\z\spr\T) \,,
\end{align*}
which together with $Y(T_\alpha a,\z)=\di_{z^\alpha} Y(a,\z)$ completes the proof.
\end{proof}

We will also investigate the case of a local action of $\cD$ on a vertex algebra in dimension $1$.

\begin{proposition}\label{ploc2}
Let\/ $(V,Y_1,\vac,T_1)$ be a vertex algebra equipped with a local action of\/ $\mathfrak{c}_D$ 
such that it annihilates the vacuum\/ $(\mathfrak{c}_D \vac = 0)$
and\/ $T_1 \in \mathfrak{c}_D$ is the translation operator of\/ $V$. Then we have $(2\le\alpha,\beta\leqslant D)$$:$
\begin{align*}
\bigl[T_{\alpha},Y_1(a,\xo)\bigr] &= Y_1(T_{\alpha}a,\xo), \qquad\;\;
\bigl[H,Y_1(a,\xo)\bigr] = Y_1(Ha,\xo) + \xo\di_\xo Y_1(a,\xo), \\
\bigl[\Om_{\alpha\beta},Y_1(a,\xo)\bigr] &= Y_1(\Om_{\alpha\beta}a,\xo), \quad
\bigl[\Om_{1\alpha},Y_1(a,\xo)\bigr] = Y_1(\Om_{1\alpha}a,\xo)
+ \xo Y_1(T_{\alpha} a,\xo), \\
\bigl[C_\alpha,Y_1(a,\xo)\bigr] &= Y_1(C_\alpha a,\xo)
+ 2 \xo Y_1(\Om_{1\alpha}a,\xo) + \xo^2 Y_1(T_\alpha a,\xo), \\
\bigl[C_1,Y_1(a,\xo)\bigr] &= Y_1(C_1 a,\xo)
- 2 \xo Y_1(Ha,\xo) - \xo^2\di_\xo Y_1(a,\xo).
\end{align*}
Conversely, the above commutation relations imply that the action of\/ $\cD$ is local and annihilates the vacuum.
\end{proposition}
\begin{proof}
We can apply \thref{tloc1} with $\z=\xo$, $T=T_1$ and any $X\in\cD$, as the action of $\ad T$ on $\cD$ is nilpotent. Using the commutation relations \eqref{cDcomm}, we find:
\begin{align*}
e^{-\xo\ad T} T_\alpha &= T_\alpha \,, &
e^{-\xo\ad T} H &= H+\xo T \,, \\
e^{-\xo\ad T} \Om_{\alpha\beta} &= \Om_{\alpha\beta} \,, &
e^{-\xo\ad T} \Om_{1\alpha} &= \Om_{1\alpha} + \xo T_\alpha \,, \\
e^{-\xo\ad T} C_\alpha &= C_\alpha +2\xo \Om_{1\alpha} + \xo^2 T_\alpha\,, &
e^{-\xo\ad T} C_1 &= C_1 - 2\xo H - \xo^2 T \,.
\end{align*}
Now the proof follows immediately from \thref{tloc1} and the property $Y_1(Ta,\xo)$ $=$ $\di_\xo Y_1(a,\xo)$.
\end{proof}

\section{The Reconstruction Theorem}\label{Se.4}
In this section, we state and prove the main result of the paper, the Reconstruction Theorem \ref{TH}.

\subsection{Formulation of the theorem}
We start by listing the necessary conditions satisfied by the action of $\mathfrak{c}_D$.

\begin{definition}\label{main-def}
Let $V$ be a 
vector space
equipped with an action of the Lie algebra $\mathfrak{c}_D$ of infinitesimal conformal transformations. 
\begin{enumerate}
\medskip
\item[$(a)$]
The action of $\mathfrak{c}_D$ is called \textbf{integrable} if every vector is contained in a {\rm finite-dimensional} subspace invariant under all infinitesimal rotations $\Om_{\alpha\beta}\in\so(D,\C)$
(as in the definition of a conformal vertex algebra in Sect.\ \ref{Se.2.3}).

\medskip
\item[$(b)$]
The action of $\mathfrak{c}_D$ is called \textbf{strongly integrable} if it is integrable and the infinitesimal dilation operator $H$ is {\rm diagonalizable} such that the spectrum of 
\beq\label{tot-cart}
H+\ii(\Omega_{12}+\Omega_{34}+\cdots+\Omega_{k,k+1})
\eeq 
consists of {\rm even integers}, where $k$ is the integer part of $\frac{D}{2}$
(cf.\ Remark $\ref{v2-rem1}$ below).
In other words,
\beq\label{tot-cart2}
(-1)^{H+\ii(\Omega_{12}+\Omega_{34}+\cdots+\Omega_{k,k+1})}
\,=\, 1 \quad \text{on} \quad V.
\eeq
$($Note that in the $\so(D+2,\C)$ realization of $\mathfrak{c}_D$, the generator $H$ is 
$\ii \Omega_{-1,0}$ 
and then in $(\ref{tot-cart})$ we have the sum over all of the Cartan basis.$)$

\medskip
\item[$(c)$]
We say that $V$ has a \textbf{positive energy} if the infinitesimal dilation operator $H$ is diagonalizable with eigenvalues in $\frac12\ZP$
(again as in the definition of a conformal vertex algebra in Sect.\ \ref{Se.2.3}).
\end{enumerate}
\end{definition}

\begin{remark}\label{v2-rem1}
The meaning of the above strong integrability condition is that 
in the even (bosonic) case, we will have a unitary representation of the geometric group of rotations and dilations, which is $\SO(D,\R) \times_{\Z/2\Z} \SO(2,\R)$, where the $\Z/2\Z$--quotient identifies $-1$ from $\SO(D,\R)$ and from the  group $\SO(2,\R)$ of dilations.
This identification will be crucial for Corollary \ref{mu-parity-cr} below, where  we will establish the poles' integrality.
In the general case (including fermions), the $\SO(D,\R)$ group is replaced by the Euclidean spinor group $\mathrm{Spin}(D,\R)$.
This is natural to expect, since unitary vertex algebras are in one-to-one correspondence with Wightman quantum field theories with global conformal invariance \cite{N05}, and in globally conformal invariant quantum field theories the spacial unitary symmetry is provided by the (connected) geometric conformal group $\SO_0 (D,2)$ (again in the bosonic case, while in the general case $\mathrm{Spin}_0 (D,2)$).
\end{remark}

The following is the main result of the paper.

\begin{theorem}\label{TH}
Let $D$ be an even positive integer.
Let\/ $(V,Y_1,\vac,T_1)$ be a vertex algebra
equipped with a local, strongly integrable, positive-energy unitary action of\/ $\mathfrak{c}_D$ such that\/ $\mathfrak{c}_D \vac = 0$ and\/ $T_1 \in \mathfrak{c}_D$ is the translation operator of\/ $V$.
Then\/ $V$ can be endowed with a unique structure\/
$(V,Y_D,\vac,\mathfrak{c}_D)$
of a conformal vertex algebra
in dimension\/ $D$, so that\/ $Y_1$ is the restriction of\/ $Y$.
\end{theorem}

The {\it proof} of this theorem is contained in the rest of this section.
Here is a brief outline.
We start in Sect.\ \ref{Se.3.1} with the reconstruction problem for $Y_D$.
Initially, we construct $Y_D(a,\z)\,b$ as a generalized series, possibly containing not only infinitely many negative powers of $\z^2$ but also half-integer powers.
Then in Sects.\ \ref{sse5.4n} and \ref{sse5.6n}, we show that in $Y_D(a,\z)\,b$ there are only integer powers of $\z^2$ and the powers are bounded from below.
Between these two subsections, in Sect.\ \ref{sse5.2n-v2}, we consider the special case $D=2$, which is not only an illustration of the reconstruction scheme of Sect.\ \ref{Se.3.1} but also serves as an important technical preparation for the 
pole bounds in Sect.\ \ref{sse5.6n}.
Then, in Sects.\ \ref{sse5.5n} and \ref{sse5.7n}, we establish all remaining axioms for the vertex algebra $(V,Y_D,\vac,\mathfrak{c}_D)$.
Our construction is more geometric,
while in
Appendix \ref{sse5.3n}, 
we give an alternative representation-theoretic interpretation of the construction of $Y_D$.

\subsection{Reconstructing the state-field correspondence}\label{Se.3.1}

Suppose that $V$ is a vertex algebra endowed with a strongly integrable local action of the conformal Lie algebra $\mathfrak{c}_D$.
Let us write
\beq\label{vert44}
Y_1 (a,\xo) b \, = \, \sum_{n \,=\, -N_{a,b}^1}^{\infty} 
\mu^1_n (a \otimes b) \xo^n \,,
\eeq
so that $\mu^1_n (a \otimes b)$ $=$ $a_{(-n-1)}b$ in the usual notation of Eq.\ (\ref{vert4}).
If we introduce the $\frac{1}{2}\ZP$--grading of $V$ provided by $H$,
\beq\label{v2e-grading1}
V \,=\, \mathop{\bigoplus}\limits_{\Delta \,\in\, \frac{1}{2}\ZP}
\, V_{\Delta}
\,,
\eeq
it follows that 
\beq\label{mu1-graded}
\mu^1_n \colon V_{\Delta'} \otimes V_{\Delta''} \to V_{\Delta'+\Delta''+n}
\,,
\eeq
i.e., $\mu^1_n$ is a degree $n$ map ($n \in \Z$).
(Note that, in general, the map $\mu^1_n$ is a $\Z$--graded map $V \otimes V \to V$ but by a slight abuse of notation we will use the same notation for its restrictions (\ref{mu1-graded}), which implicitly depend on $\Delta'$ and $\Delta''$.)

Similarly, if we had a vertex algebra structure $Y_D$ on $V$, then we could write
\beq\label{vert4-nn1}
Y_D (a,\z) b \, = \, 
\sum_{n\in\Z} \
\mathop{\sum}\limits_{m = 0}^{\infty} \
\mathop{\sum}\limits_{\sigma = 1}^{\mathfrak{h}_m^D} \
\mu^D_{n;m,\sigma}(a \otimes b) \ (\z^2)^{\frac{n-m}{2}} \, h_{m,\sigma}(\z)
\,,
\eeq
where the relation to the expansion (\ref{vert4-nn}) is via
\beq\label{e03922-1604}
\mu_{2n+m;m,\sigma}^D (a\otimes b) \,=\, a_{\{n,m,\sigma\}}b
\,.
\eeq
Furthermore, the commutation relations (\ref{CVA1}) with $H$ require that
\beq\label{muD-graded}
\mu^D_{n;m,\sigma} \colon V_{\Delta'} \otimes V_{\Delta''} \to V_{\Delta'+\Delta''+n}
\,.
\eeq
On the other hand, relations (\ref{CVA2}) are equivalent to the condition that all linear maps
\beq\label{muD-hat}
\C_m^\harm[\z] \to \Hom(V_{\Delta'} \otimes V_{\Delta''}, V_{\Delta'+\Delta''+n}) \,, \qquad
h_{m,\sigma}(\z) \mapsto \mu^D_{n;m,\sigma}
\,,
\eeq
are homomorphisms of $\so(D,\C)$--modules, where $\Om_{\alpha\beta}$ acts on $\C[\z]$
(and by restriction on $\C_m^\harm[\z]$) as $z^\alpha \di_{z^\beta} - z^\beta \di_{z^\alpha}$.

\begin{remark}\label{newrem150922}
The above maps 
$\mu^1_n \colon V \otimes V \to V$ and 
$\mu^D_{n;m,\sigma} \colon V \otimes V \to V$
(Eqs.\ (\ref{mu1-graded}), (\ref{muD-graded}))
can be expressed in terms of the residue map $\mathrm{res}_x \colon V \llb x\rrb_{x} \to V$ and its higher-dimensional analog $\mathrm{Res}_{\z} \colon V \llb \z\rrb_{\z^2} \to V$ introduced in \cite[Sect. III]{BN06}:
\beqs
\mu^1_n(a \otimes b) 
& \hspace{-8pt}
=\,
a_{(-n-1)}b
\,=
\hspace{-8pt} &
\mathrm{res}_x \, x^{-n-1} \, Y_1(a,x)b
\,,
\\
\mu^D_{2n+m;m,\sigma}(a \otimes b) 
& \hspace{-8pt}
=\,
a_{\{n,m,\sigma\}}b
\,=
\hspace{-8pt} &
\mathrm{Res}_{\z} \, (\z^2)^{-(D/2)-m-n} \, h_{m,\sigma} (\z) \, Y_D(a,\z)b
\,
\eeqs
(cf.\ \cite[Eq.\ (3.8)]{BN06}).
This brings further clarity to the $\mathfrak{so} (D,\C)$--equivariance of the map (\ref{muD-hat}) in a conformal vertex algebra in dimension $D$.
\end{remark}

In this subsection, we give an explicit construction of the maps $\mu^D_{n;m,\sigma}$ 
from the state-field correspondence $Y_1$ and the action of $\cD$,
satisfying the conditions of Theorem \ref{TH}.
In the subsequent subsections, we will derive the properties of $Y_D$. In particular,
in Sect.~\ref{sse5.4n} we prove that the powers of $\z^2$ that appear in Eq.~(\ref{vert4-nn1}) are integral;
in Sect.~\ref{sse5.6n} we prove that these powers are bounded from bellow;
in Sect.~\ref{sse5.5n} we prove that the so-obtained series (\ref{vert4-nn1}) is the only $\so(D,\C)$--equivariant series whose restriction to $\z$ $=$ $(x,$ $0,$ $\dots,$ $0)$ gives (\ref{vert44}).

For the sake of simplicity, from now on we shall make the additional assumption that the eigenspaces $V_{\Delta}$ of $H$ are {\it finite dimensional}.
In this way, the linear maps
$\mu^1_n$ and $\mu^D_{n;m,\sigma}$ 
become maps between finite-dimensional spaces (see \eqref{mu1-graded}, \eqref{muD-graded}). 
However, this assumption is not crucial, since due to the integrability condition $(a)$ from Definition \ref{main-def}, for any two elements $a \in V_{\Delta'}$ and $b \in V_{\Delta''}$, one can find finite-dimensional subspaces 
$W' \subseteq V_{\Delta'}$ and $W'' \subseteq V_{\Delta''}$ with $a\in W'$, $b\in W''$, which are invariant under the $\so(D,\C)$--subalgebra.
Then the image $W := \mu^1_n (W' \otimes W'')$ will be a finite-dimensional subspace of $V_{\Delta'+\Delta''+n}$, which is also invariant under the $\so(D,\C)$--subalgebra.
The maps $\mu^D_{n;m,\sigma}$ will be constructed again as linear maps $W'\otimes W'' \to W$. 
Furthermore, the resulting $\mu^D_{n;m,\sigma}(a \otimes b)$ will not depend on the choice of $W'$ and $W''$. 
Indeed, if we choose other subspaces $U'$ and $U''$ with $a \in U'$, $b \in U''$, then without loss of generality we can assume that $U' \subseteq W'$ and $U'' \subseteq W''$; hence setting $U := \mu^1_n (U' \otimes U'') \subseteq W$ we can restrict the construction from $W'\otimes W'' \to W$ to $U'\otimes U'' \to U$ due to the $\so(D,\C)$--equivariance of $\mu^1_n$.
However, the proofs bellow become more transparent if we assume that all $V_{\Delta}$ are finite dimensional, which in fact is the most relevant case in physics.

We pass now to the construction of the maps $\mu^D_{n;m,\sigma} \colon V_{\Delta'} \otimes V_{\Delta''} \to V_{\Delta'+\Delta''+n}$. 
Let us consider a complex vector $\u \in \C^D$ with $\u^2 = \u \spr \u=1$ (not a formal variable). We can also write it as
$$
\u \,=\, (\cos \vartheta, \sin \vartheta \, \u_{\perp}')
$$
with $\vartheta \in \C$ and $\u_{\perp}' \in \C^{D-1}$ where again $(\u_{\perp}')^2=1$.
Define a group element $g_{\u}$ $\in$ $\SO(D,\C)$ by
$$
g_{\u} := e^{\vartheta \Omega_{1,\u'_{\perp}}}
\,,\qquad
\Omega_{1,\u'_{\perp}} :=
\mathop{\sum}\limits_{\alpha = 2}^{D} \Omega_{1\alpha} {u'}_{\hspace{-2pt}\perp}^{\alpha}
\,.
$$
Then $g_{\u}$ has the following properties:
\begin{enumerate}
\medskip
\item[$(i)$] $g_{\u} (\e_1)$ $=$ $\u$ \; where \; $\e_1=(1,0,\dots,0) \in\C^D$.
\medskip
\item[$(ii)$] Let $U(g_{\u})$ be the representation of $g_{\u}$ on $V$ after the integration of the $\so(D,\C)$--action (according to the integrability condition of Definition~\ref{main-def}$(a)$).
Then each $V_{\Delta}$ is invariant under $U(g_{\u})$, and in a basis these are represented by matrices whose elements are polynomials in $\u$.
\end{enumerate}
\medskip

We then define
\beq\label{mud-def}
\mathop{\sum}\limits_{m = 0}^{\infty} \
\mathop{\sum}\limits_{\sigma = 1}^{\mathfrak{h}_m^D} \ 
\mu_{n;m,\sigma}^D (a \otimes b) \
h_{m,\sigma} (\u)
\,:=\,
U(g_{\u}) \, \mu_n^1 \bigl(U (g_{\u})^{-1}a \otimes U (g_{\u})^{-1}b\bigr)
\,,
\eeq
the left-hand side being a finite sum as the right-hand side is a polynomial in $\u$ whose degree depends on $a,b$ and $n$.
Note that in the right-hand side of \eqref{mud-def}, we have the 
natural finite-dimensional action of $g_{\u}\in \SO(D,\C)$ on
$\mu_n^1 \in \Hom(V_{\Delta'} \otimes V_{\Delta''}, V_{\Delta'+\Delta''+n})$.
Eq.\ \eqref{mud-def} determines unique coefficients $\mu_{n;m,\sigma}^D (a \otimes b)$, because for fixed $m$ the harmonic polynomials $h_{m,\sigma} (\u)$ form a basis
of the space of homogenous polynomials of degree $m$ on the sphere $\{ \u \in \C^D \,|\, \u^2 =1\}$
(see \reref{remharm}).

\subsection{Parity property}\label{sse5.4n}
 
\begin{lemma}\label{lm5.x1}
Let\/ $f_{\u}$ $\in$ $\SO(D,\C)$ be a function of\/ $\u \in \C^D$, where $\u^2=1$, which satisfies the above two conditions $(i)$ and $(ii)$.
Then
\beq\label{gu-inv}
U(g_{\u}) \, \mu_n^1 \bigl(U (g_{\u})^{-1}a \otimes U (g_{\u})^{-1}b\bigr)
\,=\,
U(f_{\u}) \, \mu_n^1 \bigl(U (f_{\u})^{-1}a \otimes U (f_{\u})^{-1}b\bigr)
\,.
\eeq
\end{lemma}

\begin{proof}
Since $g_{\u}^{-1}f_{\u} (\e_1)$ $=$ $\e_1$, we have that $g_{\u}^{-1}f_{\u}$ $\in$ $\SO(D-1,\C)$, the stabilizer group of $\e_1$.
The Lie algebra of this subgroup is the $\so(D-1,\C)$--subalgebra of $\so(D,\C)$ spanned by $\Om_{\alpha\beta}$ for $2 \leqslant\alpha, \beta\leqslant D$.
Recall from \prref{ploc2} that 
$$
\bigl[\Om_{\alpha\beta},Y_1(a,\xo)\bigr] = Y_1(\Om_{\alpha\beta}a,\xo) 
\,, \qquad 2 \leqslant\alpha, \beta\leqslant D \,.
$$
This implies that the map $\mu^1_n$ is $\so(D-1,\C)$--equivariant, i.e., $\so(D-1,\C)$ acts trivially on 
$\mu_n^1 \in \Hom(V_{\Delta'} \otimes V_{\Delta''}, V_{\Delta'+\Delta''+n})$.
Hence, $\mu_n^1$ is $\SO(D-1,\C)$--invariant, and this implies (\ref{gu-inv}).
\end{proof}

\begin{lemma}\label{mu-sod-inv-lm}
For any $g \in \SO(D,\C)$, we have
\beq\label{mu-sod-inv-eq}
\begin{split}
U(g) &
\mathop{\sum}\limits_{m = 0}^{\infty} \
\mathop{\sum}\limits_{\sigma = 1}^{\mathfrak{h}_m^D} \ 
\mu_{n;m,\sigma}^D (U(g)^{-1}a \otimes U(g)^{-1}b) \
h_{m,\sigma} (\u) \\ 
&=
\mathop{\sum}\limits_{m = 0}^{\infty} \
\mathop{\sum}\limits_{\sigma = 1}^{\mathfrak{h}_m^D} \ 
\mu_{n;m,\sigma}^D (a \otimes b) \
h_{m,\sigma} \bigl(g(\u)\bigr) \,.
\end{split}
\eeq
\end{lemma}

\begin{proof}
We apply Lemma \ref{lm5.x1} to $f_{\u}$ $:=$ $g^{-1}g_{g(\u)}$ and use Eq.~(\ref{mud-def}).
\end{proof}

\begin{corollary}\label{mu-parity-cr}
We have\/
$\mu^D_{n;m,\sigma}=0$ if\/ $n-m$ is odd.
\end{corollary}

\begin{proof}
We need to prove that the left-hand side of (\ref{mud-def}) as a function of $\u$ has parity $(-1)^n$.
This follows if we apply the previous lemma for 
$$
g \,=\, (-1)^{\ii(\Omega_{12}+\Omega_{34}+\cdots+\Omega_{k-1,k})}
\,,\quad
k \,:=\, \frac{D}{2}
$$
(recall that $D$ is even),
and use that according to Eq.~(\ref{tot-cart2}) we have
$U(g)$ $=$ $(-1)^H$. 
\end{proof}

\subsection{The case $D=2$}\label{sse5.2n-v2}

Before continuing with the general case of Theorem \ref{TH}, it will be useful to consider the special case $D=2$.
In this case, we have a chiral decomposition of the conformal Lie algebra $\mathfrak{c}_2\cong \so(4,\C)$, 
which becomes isomorphic to a direct sum of two copies of the M\"obius Lie algebra $\mathfrak{c}_1 \cong \so(3,\C) \cong \sl(2,\C)$.
To write this decomposition explicitly, let us introduce:
\beqa\label{v2e070922-1349}
L_{-1}^{\pm} 
&& \hspace{-15pt}
:=\, 
\frac{1}{2} \bigl(T_1\mp \ii T_2\bigr)
\,,\quad
L_1^{\pm} 
\,:=\, 
\frac{1}{2} \bigl(C_1\pm \ii C_2\bigr)  \,=\, \bigl(L_{-1}^{\pm}\bigr)^*
\,,\quad
\nonumber \\
L_0^{\pm} \,
&& \hspace{-15pt}
:=\, 
\frac{1}{2} \bigl(H\mp \ii \Omega_{12}\bigr)  
\,=\, 
\bigl(L_0^{\pm}\bigr)^*
\,.
\eeqa
Then commutation relations (\ref{cDcomm}) become:
\beq\label{v2e070922-1352}
\bigl[L^{\pm}_m,\,L^{\pm}_n\bigr] \,=\, (m-n) L^{\pm}_{m+n}
\,,\quad
\bigl[L^+_m,\,L^-_n\bigr] \,=\, 0 \,,
\eeq
for $m,n \in \{-1,0,1\}$.

Introduce the chiral coordinates
$$
z^{\pm} := z^1 \pm \ii z^2 \,,\quad\text{so that} \quad \z^2 = z^+z^- \,.
$$
Then in the representation of $\mathfrak{c}_2$ on $\C[\z] = \C[z^1,z^2]=\C[z^\pm]$ by differential operators, where $T_{\alpha}$ acts as $\di_{z^{\alpha}}$ for $\alpha=1,2$, we have that 
$L_{-1}^{\pm}$ are represented by $\di_{z^{\pm}}$. In this representation, $H$ acts as 
$z^1\di_{z^1}+z^1\di_{z^2}$ 
$=$
$z^+\di_{z^+}+z^-\di_{z^-}$, 
while $\Omega_{12}$ acts as
$z^1\di_{z^2}-z^2\di_{z^1}$
$=$
$\ii (z^+\di_{z^+}-z^-\di_{z^-})$; hence,
$L_0^{\pm}$ are represented by $z^{\pm}\di_{z^{\pm}}$.

Next, the integrability condition $(a)$ of Definition \ref{main-def} just means that $\ii \Omega_{12}$ is diagonalizable, while the strong integrability condition $(b)$ further implies that $V$ decomposes into a direct sum of eigenspaces:
\beq\label{v2e-grading2}
V \,=\, \mathop{\bigoplus}\limits_{(\Delta^+,\Delta^-) \,\in\, \ZP \times \ZP}
\, V_{\Delta^+,\Delta^-}
\,,\qquad
L_0^{\pm} \Bigl|_{V_{\Delta^+,\Delta^-}}
\,=\,
\Delta^{\pm} \, \mathit{id}_{V_{\Delta^+,\Delta^-}}
\,,
\eeq
where 
we use notation 
similar to Eq.\ (\ref{v2e-grading1}).
In particular, $L_0^{\pm}$ are simultaneously diagonalizable on $V$ and have integral eigenvalues.

Now, let us see how Eq.\ (\ref{mud-def}) works in this case.
Note that in the chiral coordinates $z^{\pm}$, there is a natural basis of homogeneous harmonic polynomials of degree $m$: 
\beq\label{hmpm}
h_{m,\pm} (z^+,z^-) \,:=\, (z^{\pm})^m
\quad  \text{for} \quad m > 0 \,, \qquad
h_0 (z^+,z^-) \,:=\, 1 
\,.
\eeq
It is convenient to set
$$
h_{m} := h_{m,+} \,, \quad h_{-m} := h_{m,-} \,, \qquad m > 0 \,.
$$
Then, for $\u$ $=$ $\z/(\z^2)^{\frac12}$, we have that $g_{\u}$ $=$ $e^{\ii \vartheta \Omega_{12}}$ for 
$\cos \vartheta = z^1/(\z^2)^{\frac12}$, i.e., 
$e^{2\ii\vartheta} = z^+/z^- = (u^{\pm})^{\pm 2}$.
Hence, 
$$
h_m (\u) \,=\, \Bigl(\frac{z^+}{z^-}\Bigr)^{\frac{m}{2}}
\,,\qquad m \in \Z \,,
$$
and Eq.\ (\ref{mud-def}) reads
$$
\mathop{\sum}\limits_{m \,\in\, \Z} \ 
\mu_{n;m}^2 (a \otimes b) \
\Bigl(\frac{z^+}{z^-}\Bigr)^{\frac{m}{2}}
\,=\,
\Bigl(\frac{z^+}{z^-}\Bigr)^{\frac{1}{2}(L_0^+-\Delta'{}^+-\Delta''{}^+-L_0^-+\Delta'{}^-+\Delta''{}^-)}
\mu_n^1 (a \otimes b)
\,
$$
for 
$a \in V_{\Delta'{}^+,\Delta'{}^-}$ 
and
$b \in  V_{\Delta''{}^+,\Delta''{}^-}$.
In other words, if we denote by $\mathrm{Pr}_{\Delta^+,\Delta^-}$ the projection onto $V_{\Delta^+,\Delta^-}$ in the direct sum (\ref{v2e-grading2}), then
$$
\mu^2_{n;m} (a \otimes b)
\,=\,
\mathrm{Pr}_{\Delta'{}^++\Delta''{}^++n^+,\Delta'{}^-+\Delta''{}^-+n^-}
\bigl(
\mu^1_n(a \otimes b) 
\bigr),
$$
where $n^{\pm} := (n\pm m)/2$, i.e., $m = n^+-n^-$ and $n = n^++n^-$.

\subsection{Pole bounds}\label{sse5.6n}

In this subsection, we shall specifically single out $\ii\Omega_{12}$ as in the previous subsection and diagonalize it simultaneously with $H$.

\begin{lemma}\label{v2lm-5.5-new}
Let\/ $V$ be a unitary representation of the conformal Lie algebra $\mathfrak{c}_D$, which is strongly integrable and has positive energy 
$($in the sense of Definition $\ref{main-def}$$(b)$, $(c)$$)$.
If\/ $(\Delta,j)$ is a pair of common eigenvalues for $H$ and $\ii\Omega_{12}$, respectively, then it obeys the inequality $\Delta \geqslant |j|$.
\end{lemma}

\begin{proof}
Since $V$ is a unitary positive energy representation of $\mathfrak{c}_D$, it is generated as a $\C[T_1,\dots,T_D]$--module by the subspace of quasi-primary elements, i.e., those $a \in V$ for which $C_{\alpha} a = 0$ for all $\alpha=1,\dots,D$.
We shall prove the lemma first for these quasi-primary elements.

Let us consider the subrepresentation of the algebra $\mathfrak{c}_2$ and introduce the M\"obius generators of Eqs.\ (\ref{v2e070922-1349}), (\ref{v2e070922-1352}).
If $a$ is a non-zero quasi-primary element that has eigenvalues $\Delta^{\pm}$ for $L_0^{\pm}$, respectively, then we have
$$
0 \,\leqslant\,
\|L_{-1}^{\pm}a\|^2 \,=\,
\langle L_{-1}^{\pm}a | L_{-1}^{\pm}a \rangle
\,=\,
\langle a | L_1^{\pm}L_{-1}^{\pm}a \rangle
\,=\,
2\Delta^{\pm} \|a\|^2
\,,
$$
since $L_1^{\pm}a = 0$.
Hence, $\Delta^{\pm} = (\Delta \mp j)/2$ are both non-negative.

A general eigenvector $a \in V$ with eigenvalues $(\Delta,j)$
can be obtained as a linear combination of products 
$X_1 \cdots X_{\ell}$ acting on quasi-primary vectors, where 
$X_1,\dots,X_{\ell} \in \{L_{-1}^+, L_{-1}^-, T_3 ,\dots, T_D\}$.
Then the bound $\Delta \geqslant |j|$ is established by a straightforward induction on $\ell$.
\end{proof}

\begin{lemma}\label{Lm.5.6-1new}
For every\/ $a, b \in V$, there exists\/ $N^D_{a,b}$ $\in$ $\ZP$ such that\/ $\mu^D_{n;m,\sigma} (a \otimes b)=0$ for all possible\/ $m$ and\/ $\sigma$ if\/ $\frac{m-n}{2}>N^D_{a,b}$ $($cf.\ Eqs.\ \eqref{vert4-nn1}, \eqref{e03922-1604}$)$.
\end{lemma}

\begin{proof}
We prove the contrapositive statement: assuming that some
$\mu^D_{n;m,\sigma_0} (a$ $\otimes$ $b)\neq 0$
we shall find an upper bound $N^D_{a,b}$ for
$$
N 
\,:=\,
\frac{m-n}{2}
\,,
$$
which depends only on $a$ and $b$. 
We can find irreducible $\so(D,\C)$--sub\-rep\-re\-sen\-ta\-tions $W' \subseteq V_{\Delta'}$ and $W'' \subseteq V_{\Delta''}$ such that $a \in W'$ and $b \in W''$. We set
$$
W \,:=\, \mu^1_n(W' \otimes W'') \,\subseteq\, V_{\Delta'+\Delta''+n}
\,,
$$
which is again an $\so(D,\C)$--subrepresentation.

Recall that the element $\Omega_{\alpha\beta}\in\so(D,\C)$ acts on $\C[\z]$ as $z^\alpha\di_{z^\beta}-z^\beta\di_{z^\alpha}$.
As $\ii\Omega_{12}$ is Hermitian (see \eqref{eq2.1}), we can choose the basis $\{h_{m,\sigma}(\z)\}_{\sigma = 1}^{\mathfrak{h}_m^D}$ of $\C_m^\harm[\z]$
so that it consists of eigenvectors for $\ii\Omega_{12}$. 
One of the eigenvalues is $m$, and we let it correspond to the index $\sigma_m$.
In fact, explicitly we can set (cf.\ \eqref{hmpm}):
$$
h_{m,\sigma_m} (\z) \,:=\, (z^1 - \ii z^2)^m \,,
$$
and then we have
$$
\ii\Omega_{12} \, h_{m,\sigma_m} (\z)
\,=\,
m \,h_{m,\sigma_m} (\z)
\,.
$$

As a consequence of Lemma \ref{mu-sod-inv-lm}, the linear map
\beq\label{so-map}
\C^{\harm}_m[\z] \to \Hom(W' \otimes W'', W) \,, \qquad
h_{m,\sigma} (\z) \mapsto \mu^D_{n;m,\sigma} \,,
\eeq
is a homomorphism of $\so(D,\C)$--modules (cf.\ \eqref{muD-hat}).
This map is non-zero, because by assumption $\mu^D_{n;m,\sigma_0} \neq 0$.
Since $\C^{\harm}_m[\z]$ is an irreducible $\so(D,\C)$--module (cf.\ Appendix \ref{sse5.3n}), it follows from Schur's Lemma that the map \eqref{so-map} is injective.

In particular, there exist $c \in W'$ and $d \in W''$ such that
$$
\mu^D_{n;m,\sigma_m} (c \otimes d) \,\neq\, 0
\,.
$$
Without loss of generality, we can assume that $c$ and $d$ are eigenvectors for $\ii\Omega_{12}$ with eigenvalues $k'$ and $k''$, respectively.
Hence,
$\mu^D_{n;m,\sigma_m}(c \otimes d)$ has eigenvalues $(\Delta'+\Delta''+n,k'+k''+m)$ for $H$ and $\ii\Omega_{12}$, respectively.
By Lemma \ref{Lm.5.6-1new}, we have that
$$
\Delta'+\Delta''+n \,\geqslant\,k'+k''+m
$$
and so,
$m-n \leqslant \Delta'+\Delta''-k'-k''$,
which finally gives us
\beq\label{newpolebound-070922-0925}
N \,\leqslant\,
\frac{1}{2}\bigl(\Delta'+J'+\Delta''+J''\bigr)
\,,
\eeq
where $-J'$ and $-J''$ are the minimal eigenvalues of $\ii\Omega_{12}$ on $W'$ and $W''$, respectively.
This proves the lemma.
\end{proof}

\begin{remark}\label{rm080922-0941}
In the physically most interesting case of $D=4$, there is a local isomorphism between $\SO(4,\C)$ and $\SO(3,\C) \times \SO(3,\C)$. 
Suppose that in the above proof the weights of $W'$, $W''$ are 
$$
(j_1',j_2')\,,\quad
(j_1'',j_2'')\,,\quad
$$
respectively, as irreducible $\so(3,\C) \oplus \so(3,\C)$--representations. Then $J'=j_1'+j_2'$ and $J''=j_1''+j_2''$, and estimate (\ref{newpolebound-070922-0925}) becomes:
$$
N 
\,\leqslant\,
\frac{1}{2}\,
\bigl(
\Delta'+j_1'+j_2'+\Delta''+j_1''+j_2''
\bigr)
\,,
$$
which corresponds to an earlier result of \cite[Proposition 4.3]{NT01}
(there the estimate is a bit stronger but is derived only for $D=4$).
\end{remark}

Now, we define $Y_D(a,\z)b$ via Eqs.\ (\ref{vert4-nn1}), (\ref{mud-def}), and as a direct corollary of the above lemma,
$Y_D(a,\z)b$ becomes a series from the space $V \llb \z \rrb_{\z^2}$.
Since $g_{\u}$ $=$ $1$ for $\u$ $=$ $\e_1$, it follows that
$$
\mathop{\sum}\limits_{m = 0}^{\infty} \
\mathop{\sum}\limits_{\sigma = 1}^{\mathfrak{h}_m^D} \ 
\mu_{n;m,\sigma}^D (a \otimes b) \
h_{m,\sigma} (\u)
\Bigl|_{\u = \e_1} \,=\, \mu_n^1 (a \otimes b) \,,
$$
and hence
\beq\label{D1restr}
Y_D(a,\z)b \bigl|_{\z = x\e_1} \,=\, Y_1(a,x)b
\,.
\eeq
In particular, this implies that the reconstructed state-field correspondence $Y_D$ is non-zero, because $Y_1\ne0$ by the definition of a vertex algebra.

From Eqs.\ (\ref{vert4-nn1}), (\ref{mud-def}), we derive the more explicit formula:
\beqa\label{vert4-nn1-mod090422}
Y_D (a,\z) b 
&&\hspace{-15pt} 
= \, 
\sum_{n \,\geqslant\, -N^D_{a,b}} \
\mathop{\sum}\limits_{m = 0}^{\infty} \
\mathop{\sum}\limits_{\sigma = 1}^{\mathfrak{h}_m^D} \
\mu^D_{n;m,\sigma}(a \otimes b) \ (\z^2)^{\frac{n}{2}} \, h_{m,\sigma}(\u)
\nonumber \\
&&\hspace{-15pt} 
= \, 
\sum_{n \,\geqslant\, -N^D_{a,b}} \
U(g_{\u(\z)}) \, \mu_n^1 \bigl(U (g_{\u(\z)})^{-1}a \otimes U (g_{\u(\z)})^{-1}b\bigr) \ (\z^2)^{\frac{n}{2}}
\nonumber \\
&&\hspace{-15pt} 
= \, 
U(g_{\u(\z)}) \,  Y_1\bigl( U(g_{\u(\z)})^{-1}a, (\z^2)^{\frac{1}{2}}\bigr)\,U(g_{\u(\z)})^{-1} b
\eeqa
for $\u(\z) = {\z}/(\z^2)^{\frac{1}{2}}$, where the second and the third lines can be viewed as generating series for the definition of $\mu^D_{n;m,\sigma}$ in the first line.
Note that if we restrict $Y_D(a,\z)b$ to $a \in V_{\Delta' }$ and $b \in V_{\Delta''}$ and then project the result onto $V_{\Delta}$ in the direct sum
$V = \mathop{\bigoplus}\limits_{\Delta \,\in\, \frac{1}{2}\ZP} \,V_{\Delta}$,
then it will give rise to a map between finite-dimensional spaces whose matrix entries are polynomials in $\z$ and $(\z^2)^{-\frac{1}{2}}$, but due to the above results all half-integer powers of $\z^2$ will be cancelled.

\subsection{Covariance properties}\label{sse5.5n}

\begin{lemma}\label{SOD-covar}
The series\/ $Y_D (a,\z) b$, constructed by Eqs.~\eqref{vert4-nn1}, \eqref{mud-def}, is\/ $\SO(D,$ $\C)$--equivariant, and is the unique\/ $\SO(D,\C)$--equivariant series whose restriction to\/ $\z$ $=$ $x\e_1$ is equal to\/ $Y_1 (a,x)b$.
\end{lemma}

\begin{proof}
The $\SO(D,\C)$--equivariance of $Y_D (a,\z) b$ is equivalent to the $\SO(D,\C)$--equi\-var\-iance of the left-hand side of Eq.~(\ref{mud-def}),
which was established in Lemma \ref{mu-sod-inv-lm}. To see the uniqueness of $Y_D (a,\z) b$, note that
the $\SO(D,\C)$--equivariance of the left-hand side of Eq.~(\ref{mud-def}) and the fact that its restriction to $\u$ $=$ $\e_1$ is $Y_1(a,x)b$ implies Eq.~(\ref{mud-def}).
\end{proof}

\begin{lemma}\label{lm090422-1429}
The translation operators $T_{\alpha}$ are derivations of all maps $\mu^{D}_{n;m,\sigma}$, i.e., we have
$$
T_{\alpha} \,\mu^{D}_{n;m,\sigma} (a\otimes b)
\,=\,
\mu^{D}_{n;m,\sigma} (T_{\alpha}a\otimes b)
+
\mu^{D}_{n;m,\sigma} (a\otimes T_{\alpha} b)
\,,
$$
for all $a,b \in V$ and all possible values of\/ $\alpha$, $n$, $m$ and $\sigma$.
Equivalently, we have
\beq\label{e090422-1624}
\bigl[T_{\alpha}, Y_D(a,\z)\bigr]
\,=\,
Y_D(T_{\alpha}a,\z)
\,.
\eeq

\end{lemma}

\begin{proof}
Consider the linear combination $T(\vv) := \mathop{\sum}\limits_{\beta = 1}^{D} v^{\beta} T_{\beta}$ for an arbitrary 
complex vector $\vv = (v^1,\dots,v^D) \in \C^D$. We need to prove that all $T(\vv)$ are derivations of $\mu^{D}_{n;m,\sigma}$, and
we already know that $T(\vv)$ are derivations of $\mu^1_n$ for every $\vv$ and $n$.

Note that $T(\vv)\,U(g_{\u}) = U(g_{\u})\,T(\vv')$ where $\vv' := g_{\u}^{-1}(\vv)$.
Then we apply Eq.\ (\ref{mud-def}) and compute:
\begin{align*}
\mathop{\sum}\limits_{m = 0}^{\infty} \
& \mathop{\sum}\limits_{\sigma = 1}^{\mathfrak{h}_m^D} \ 
\Bigl(T(\vv) \,
\mu_{n;m,\sigma}^D (a \otimes b) \Bigr)\,
h_{m,\sigma} (\u)
\\ &
=\,
T(\vv) \,
U(g_{\u}) \, \mu_n^1 \bigl(U (g_{\u})^{-1}a \otimes U (g_{\u})^{-1}b\bigr)
\\ &
=\,
U(g_{\u}) \, T(\vv') \, \mu_n^1 \bigl(U (g_{\u})^{-1}a \otimes U (g_{\u})^{-1}b\bigr)
\\ &
=\,
U(g_{\u}) \, \mu_n^1 \bigl(U (g_{\u})^{-1} \, T(\vv) \, a \otimes U (g_{\u})^{-1}b\bigr)
\\ &
\quad+\
U(g_{\u}) \, \mu_n^1 \bigl(U (g_{\u})^{-1} \, a \otimes U (g_{\u})^{-1} T(\vv) \, b\bigr)
\\ &
=\,
\mathop{\sum}\limits_{m = 0}^{\infty} \
\mathop{\sum}\limits_{\sigma = 1}^{\mathfrak{h}_m^D} \ 
\Bigl(
\mu_{n;m,\sigma}^D (T(\vv) \, a \otimes b) 
+
\mu_{n;m,\sigma}^D (a \otimes T(\vv) \, b) 
\Bigr)\,
h_{m,\sigma} (\u)
\,,
\end{align*}
which completes the proof of the lemma.
\end{proof}

\begin{lemma}\label{lm090422-1725}
We have 
\beq\label{e090422-1726}
Y_D(T_{\alpha}a,\z) \, b 
\,=\, 
\di_{z^{\alpha}} Y_D(a,\z) \, b
\eeq
for all $a,b \in V$ and $\alpha=1,\dots,D$.
\end{lemma}

\begin{proof}
Let us differentiate with respect to $z^{\alpha}$ the third line of Eq.\ (\ref{vert4-nn1-mod090422}) of the generating series of $Y_D (a,\z)b$:
\beqa\label{e050922-1059}
\di_{z^{\alpha}} \, Y_D(a,\z) \, b
\hspace{-15pt}
&&
=\,\di_{z^{\alpha}}
\Bigl(
U(g_{\u(\z)}) \,  Y_1\bigl( U(g_{\u(\z)})^{-1}\,a, (\z^2)^{\frac{1}{2}}\bigr)\,U(g_{\u(\z)})^{-1} b
\Bigr)
\nonumber \\ &&
=\
\di_{z^{\alpha}} \bigl(U(g_{\u(\z)})\bigr)  \,  Y_1\bigl( U(g_{\u(\z)})^{-1}\,a, (\z^2)^{\frac{1}{2}}\bigr)\,U(g_{\u(\z)})^{-1} b
\nonumber \\ &&
+\
U(g_{\u(\z)}) \,  Y_1\bigl( \di_{z^{\alpha}} \bigl(U(g_{\u(\z)})^{-1}\bigr)\,a, (\z^2)^{\frac{1}{2}}\bigr)\,U(g_{\u(\z)})^{-1} b
\\ \nonumber  &&
+\
U(g_{\u(\z)}) \,  Y_1\bigl( U(g_{\u(\z)})^{-1}\,a, (\z^2)^{\frac{1}{2}}\bigr)\,\di_{z^{\alpha}} \bigl(U(g_{\u(\z)})^{-1}\bigr) b
\\ \nonumber  &&
+\
U(g_{\u(\z)}) \, 
\Bigl(z^{\alpha}\,x^{-1}\di_x
Y_1\bigl( U(g_{\u(\z)})^{-1}\,a,x\bigr)
\Bigr)\Bigl|_{x\,:=\,(\z^2)^{\frac{1}{2}}} \,U(g_{\u(\z)})^{-1} b
\,.
\eeqa
For the last term in the right-hand side, we apply the equation $\di_x Y_1(a',x)b' = Y_1 (T_1\,a',x)b'$, which holds
for arbitrary $a',b' \in V$. 
For the previous three terms, we first calculate:
\begin{align*}
\di_{z^{\alpha}} \bigl(U(g_{\u(\z)})\bigr)
&=\,
U(g_{\u(\z)}) \,\di_{{z'}^{\alpha}} 
\bigl(U(g_{\u(\z)}^{-1}\,g_{\u(\z')})\bigr)\bigl|_{\z' \,=\, \z}
\,=:\, U(g_{\u(\z)}) \, \widetilde{\Omega} 
\,, \\ 
\di_{z^{\alpha}} \bigl(U(g_{\u(\z)})^{-1}\bigr)
&=\,
- \widetilde{\Omega} \, U(g_{\u(\z)})^{-1}
\,,
\end{align*}
where $\widetilde{\Omega}$ is a linear combination of $\Omega_{\alpha\beta}$ with coefficients that are polynomial in $\z$ and $(\z^2)^{-\frac{1}{2}}$.
Then we apply the commutation relations with $\Omega_{\alpha\beta}$ from Proposition \ref{ploc2} to obtain:
$$
\widetilde{\Omega} \, Y_1(a',x)\,b' 
-
Y_1(\widetilde{\Omega}\,a',x)\,b'
-
Y_1(a',x)\,\widetilde{\Omega}\,b'
\,=\,
Y_1(\widetilde{T}\,a',x)
\,	
$$
for any $a',b' \in V$,
where $\widetilde{T}$ is a linear combination of $T_{\alpha}$ again with polynomial coefficients in $\z$ and $(\z^2)^{-\frac{1}{2}}$.
Thus, the right-hand side of (\ref{e050922-1059}) becomes:
\begin{align*}
\di_{z^{\alpha}} & Y_D(a,\z) \, b
\nonumber \\ &
=\, U(g_{\u(\z)}) \,  Y_1\bigl( \widetilde{T} \, U(g_{\u(\z)})^{-1}\,a, (\z^2)^{\frac{1}{2}}\bigr)\,U(g_{\u(\z)})^{-1} b
\nonumber \\ &
\quad+\
U(g_{\u(\z)})  \, z^{\alpha} (\z^2)^{-\frac{1}{2}} Y_1\bigl( T_1 \, U(g_{\u(\z)})^{-1}\,a, (\z^2)^{\frac{1}{2}}\bigr)\,U(g_{\u(\z)})^{-1} b
\nonumber \\ &
=\, U(g_{\u(\z)}) \,  Y_1\bigl( U(g_{\u(\z)})^{-1}\, T_{\alpha}\,a, (\z^2)^{\frac{1}{2}}\bigr)\,U(g_{\u(\z)})^{-1} b
\,.
\end{align*}

The second equality above is a consequence of the relation
\beq\label{CHECK-IT}
U(g_{\u(\z)}) \, \bigl(\widetilde{T} + z^{\alpha}\,(\z^2)^{-\frac{1}{2}} \, T_1 \bigr) \, U(g_{\u(\z)})^{-1}
\,=\,
T_{\alpha}
\,,
\eeq
which can be derived from an explicit computation of $\widetilde{\Omega}$ and $\widetilde{T}$.
This computation is a bit cumbersome, but fortunately there is a simplifying argument.
Since both sides of Eq.\ (\ref{e090422-1726}) are $\SO(D,\C)$--equivariant, it is enough to compare them for $\z=x\e_1$,
so we only need to establish Eq.\ \eqref{CHECK-IT} for $\z=x\e_1$.
To this end, we first find
$\widetilde{\Omega}\bigl|_{\z=x\e_1} = x^{-1}\Omega_{1\alpha}$,
and then it follows that
$$\widetilde{T}\bigl|_{\z=x\e_1} = (1-\delta_{\alpha,1})T_{\alpha}\,.$$
Therefore,
$\widetilde{T} + z^{\alpha}\,(\z^2)^{-\frac{1}{2}} T_1$
restricts to $T_{\alpha}$ for $\z=x\e_1$, as claimed.
\end{proof}

\begin{lemma}\label{lm090422-1925}
The commutation relation $(\ref{CVA1})$ for $H$ holds, i.e.,
$$
\bigl[H,Y_D(a,\z)\bigr] = Y(Ha,\z) + \z\cdot\di_\z \, Y_D(a,\z)\,.
$$
\end{lemma}

\begin{proof}
As in the  proof of the previous lemma, we use (\ref{vert4-nn1-mod090422}) and apply the Euler differential operator to the expression in the third line of (\ref{vert4-nn1-mod090422}):
\beqs
\z \cdot \di_{\z} \, Y_D(a,\z) \, b
\hspace{-15pt}
&&
=\,
\z \cdot \di_{\z} \,
\Bigl(
U(g_{\u(\z)}) \,  Y_1\bigl( U(g_{\u(\z)})^{-1}\,a, (\z^2)^{\frac{1}{2}}\bigr)\,U(g_{\u(\z)})^{-1} b
\Bigr)
\\ &&
=\,
U(g_{\u(\z)}) \, \Bigl( x\di_x Y_1\bigl( U(g_{\u(\z)})^{-1}\,a,x\bigr)\Bigr)\Bigl|_{x\,:=\,(\z^2)^{\frac{1}{2}}}
\,U(g_{\u(\z)})^{-1} b
\\ &&
=\,
U(g_{\u(\z)}) \, 
\Bigl( 
\bigl[H, Y_1 \bigl(U(g_{\u(\z)})^{-1}\,a, (\z^2)^{\frac{1}{2}}\bigr) \bigr]
\\ &&
\quad-\ 
Y_1\bigl( H\, U(g_{\u(\z)})^{-1}\,a, (\z^2)^{\frac{1}{2}}\bigr) 
\Bigr)\,U(g_{\u(\z)})^{-1}\,b
\\ &&
=\,
\Bigl( \bigl[H, Y_D(a,\z) \bigr] - Y_D\bigl( H\,a, \z) \Bigr)\,b
\,,
\eeqs
where we use that $u(\z)$ and $(\z^2)^{\frac{1}{2}}$ are homogeneous of degree $0$ and $1$, respectively, as well as that $H$ commutes with $U(g_{\u(\z)})$.
\end{proof}

\begin{lemma}\label{lm090422-1949}
The commutation relations $(\ref{CVA3})$ for $C_{\alpha}$ hold.
\end{lemma}

\begin{proof}
We have already established (\ref{CVA1}) and (\ref{CVA2}).
Under them, Eq.\ (\ref{CVA3}) is equivalent to the following:
\begin{align}
\bigl[C(\vv),&\,Y_D(a,\z)\bigr] -
Y_D\bigl(C(\vv)\,a,\z\bigr)
\nonumber \\ 
&+
\bigl[H (\z \cdot \w),Y_D(a,\z)\bigr] +
Y_D\bigl(H (\vv \cdot \z)\,a,\z\bigr)
\nonumber \\
\label{e090422-2023}
&+
\bigl[\Omega(\z,\vv),Y_D(a,\z)\bigr] +
Y_D\bigl(\Omega(\vv,\z)\,a,\z\bigr)
\,=\, 0
\,,
\end{align}
where we have set 
$$
C(\vv) := \mathop{\sum}\limits_{\alpha \,=\, 1}^D v^{\alpha} C_{\alpha}
\,,\quad
H(\vv \cdot \z) = H  \mathop{\sum}\limits_{\alpha \,=\, 1}^D v^{\alpha}z^{\alpha} 
\,,\quad
\Omega(\vv,\z) := \mathop{\sum}\limits_{\alpha,\beta \,=\, 1}^D v^{\alpha}z^{\beta} \, \Omega_{\alpha\beta}
\,.
$$
Now, Eq.\ (\ref{e090422-2023}) is  valid for $Y_1$ in the restriction $\z = z\e_1$ 
(as a similar consequence of the relations in Proposition \ref{ploc2}).
Hence, it also holds for $Y_D$ and general $\z$ with a ``rotation'' by similar methods as those used above in Lemma \ref{lm090422-1429}.
\end{proof}

As a consequence of Lemma \ref{lm090422-1725} (and Eq.\ (\ref{e090422-1726})), we have:
$$
\iota_{\w,\z} Y_D(a,\z+\w)b \, = \, Y_D(e^{\z \, \spr \T}a,\w)b \,,
$$
where the expansion $\iota_{\w,\z}$ is defined by Eq.\ \eqref{loc4}. 
Then setting $\w=\xo\e_1$, we obtain the identity
\beq\label{Prop1}
\iota_{\xo\e_1,\z} Y_D(a,\z+\xo\e_1)b 
\, = \, Y_1(e^{\z \, \spr \T}a,\xo)b 
\,,
\eeq
in which the left-hand side is defined as
\[
\iota_{\xo\e_1,\z} Y_D(a,\z+\xo\e_1)b := e^{\z \, \spr \di_{\w}} Y_D(a,\w)b \bigr|_{\w = x\e_1} \,.
\]
We remark that formula \eqref{Prop1} could be used for an alternative construction of $Y_D$ if we could prove that the right-hand side is an $\iota_{\xo\e_1,\z}$--expansion of a Laurent polynomial in $\z+x\e_1$.

\subsection{Locality}\label{sse5.7n}

Throughout this subsection, $\u,\vv,\w,\z$ will denote $D$-dimensional 
formal variables, while $x,y$ will be $1$-dimensional formal variables.
For each $a,b\in V$, fix a non-negative integer $N_{a,b}$ ($:=N_{a,b}^D$, cf.\ \leref{Lm.5.6-1new}) such that 
\beq\label{e-loc1}
(\z^2)^{N_{a,b}} Y_D(a,\z)b \in V \llb \z \rrb
\,.
\eeq
Applying $e^{\u \spr \di_\z}$ to the left side of \eqref{e-loc1} and using the translation covariance \eqref{e090422-1726}, we obtain
$$
((\u+\z)^2)^{N_{a,b}} Y_D(e^{\u \spr \T} a,\z)b \in V \llb \u+\z \rrb \subseteq V \llb \u,\z \rrb \,.
$$
Setting $\z=x\e_1$, we get that
\beq\label{e-loc2}
((\u+x\e_1)^2)^{N_{a,b}} Y_1(e^{\u \spr \T} a,x)b \in V \llb \u,x \rrb \,.
\eeq

\begin{lemma}\label{l-loc1}
For every\/ $a,b,c\in V$, we have
\beq\label{e-loc3}
\begin{split}
((\z&+x\e_1)^2)^{N_{a,c}} ((\w+y\e_1)^2)^{N_{b,c}} ((\z-\w+x\e_1-y\e_1)^2)^{N_{a,b}}  \\
&\qquad\times Y_1(e^{\z \spr \T} a,x) Y_1(e^{\w \spr \T} b,y) c \\
&= ((\z+x\e_1)^2)^{N_{a,c}} ((\w+y\e_1)^2)^{N_{b,c}} ((\z-\w+x\e_1-y\e_1)^2)^{N_{a,b}}  \\
&\qquad\times (-1)^{p_ap_b} Y_1(e^{\w \spr \T} b,y) Y_1(e^{\z \spr \T} a,x) c
\in V \llb \z,\w,x,y \rrb \,.
\end{split}
\eeq
\end{lemma}
\begin{proof}
We start from the commutator formula \eqref{vert10} or \eqref{vert12} for the vertex algebra $(V,Y_1,\vac,T_1)$,
written in the form 
$$
\bigl[Y_1(a,x), Y_1(b,y)\bigr] = (\iota_{x,y}-\iota_{y,x}) Y_1\bigl( Y_1(a,x-y)b, y\bigr)
$$
(see \cite{FLM88} and \eqref{loc4}).
First, we replace $a$ with $e^{\u \spr \T} a$ and get
$$
\bigl[Y_1(e^{\u \spr \T} a,x), Y_1(b,y)\bigr] = (\iota_{x,y}-\iota_{y,x}) Y_1\bigl( Y_1(e^{\u \spr \T} a,x-y)b, y\bigr) .
$$
Second, we multiply both sides by $((\u+x\e_1-y\e_1)^2)^{N_{a,b}}$. Notice that the $\iota$--expansions commute
with the multiplication by a polynomial. But, by \eqref{e-loc2}, we have
$$
((\u+x\e_1-y\e_1)^2)^{N_{a,b}} Y_1(e^{\u \spr \T} a,x-y)b \in V \llb \u,x-y \rrb \subseteq V \llb \u,x,y \rrb \,.
$$
The expansions $\iota_{x,y}$ and $\iota_{y,x}$ of a formal power series in $x,y$ are equal; hence their difference is $0$.
Therefore,
$$
((\u+x\e_1-y\e_1)^2)^{N_{a,b}} \bigl[Y_1(e^{\u \spr \T} a,x), Y_1(b,y)\bigr] = 0 \,.
$$

Recall that, by \prref{ploc2}, we have
$\bigl[T_{\alpha},Y_1(a,\xo)\bigr] = Y_1(T_{\alpha}a,\xo),$
which implies that all $T_\alpha$ are derivations of $Y_1$. If we apply $e^{\vv \spr \T}$ to the above equation, it will act as an automorphism of $Y_1$ and give
$$
((\u+x\e_1-y\e_1)^2)^{N_{a,b}} \bigl[Y_1(e^{(\u+\vv) \spr \T} a,x), Y_1(e^{\vv \spr \T} b,y)\bigr] = 0 \,.
$$
Now let us apply the left side to a vector $c\in V$ and multiply it by $((\u+\vv+x\e_1)^2)^{N_{a,c}} ((\vv+y\e_1)^2)^{N_{b,c}}$. We obtain
\begin{align*}
((\u+\vv&+x\e_1)^2)^{N_{a,c}} ((\vv+y\e_1)^2)^{N_{b,c}} ((\u+x\e_1-y\e_1)^2)^{N_{a,b}} \\
&\qquad\times Y_1(e^{(\u+\vv) \spr \T} a,x) Y_1(e^{\vv \spr \T} b,y) c \\
&= ((\u+\vv+x\e_1)^2)^{N_{a,c}} ((\vv+y\e_1)^2)^{N_{b,c}} ((\u+x\e_1-y\e_1)^2)^{N_{a,b}} \\
&\qquad\times (-1)^{p_ap_b} Y_1(e^{\vv \spr \T} b,y) Y_1(e^{(\u+\vv) \spr \T} a,x) c \,.
\end{align*}
Notice that, by \eqref{e-loc2}, the left-hand side has only non-negative powers of $y$ while the right-hand side has only non-negative powers of $x$.
Therefore, both sides lie in the space $V \llb \u,\vv,x,y \rrb$.
Making the substitution $\u=\z-\w$, $\vv=\w$ completes the proof of Eq.\ \eqref{e-loc3}.
\end{proof}

\begin{lemma}\label{l-loc2}
For every\/ $a,b\in V$, we have the locality condition
\beq\label{e-loc4}
((\z-\w)^2)^{N_{a,b}} \bigl[Y_D(a,\z), Y_D(b,\w)\bigr] = 0 \,.
\eeq
\end{lemma}
\begin{proof}
Using \eqref{Prop1} and the fact that $\iota$--expansions commute with the multiplication by a polynomial, we can rewrite Eq.\ \eqref{e-loc3} as follows:
\begin{align*}
\iota_{x\e_1,\z} & \, \iota_{y\e_1,\w} \, ((\z+x\e_1)^2)^{N_{a,c}} ((\w+y\e_1)^2)^{N_{b,c}} ((\z-\w+x\e_1-y\e_1)^2)^{N_{a,b}}  \\
&\qquad\times Y_D(a,\z+x\e_1) Y_D(b,\w+y\e_1) c \\
&= \iota_{x\e_1,\z} \, \iota_{y\e_1,\w} \, ((\z+x\e_1)^2)^{N_{a,c}} ((\w+y\e_1)^2)^{N_{b,c}} ((\z-\w+x\e_1-y\e_1)^2)^{N_{a,b}}  \\
&\qquad\times (-1)^{p_ap_b} Y_D(b,\w+y\e_1) Y_D(a,\z+x\e_1) c \,.
\end{align*}
But as both sides lie in $V \llb \z,\w,x,y \rrb$, the expansions are redundant, so we can drop $\iota_{x\e_1,\z} \, \iota_{y\e_1,\w}$ from the above formula.
Moreover, since both sides are in $V \llb \z,\w,x,y \rrb$, it makes sense to set $x=y=0$ in them, which gives
\begin{align*}
(\z^2)^{N_{a,c}} & (\w^2)^{N_{b,c}} ((\z-\w)^2)^{N_{a,b}} Y_D(a,\z) Y_D(b,\w) c \\
&= (\z^2)^{N_{a,c}} (\w^2)^{N_{b,c}} ((\z-\w)^2)^{N_{a,b}} (-1)^{p_ap_b} Y_D(b,\w) Y_D(a,\z) c \,.
\end{align*}
Now we can divide both sides by $(\z^2)^{N_{a,c}} (\w^2)^{N_{b,c}}$ to obtain 
\begin{align*}
((\z-\w)^2)^{N_{a,b}} Y_D(a,\z) Y_D(b,\w) c = ((\z-\w)^2)^{N_{a,b}} (-1)^{p_ap_b} Y_D(b,\w) Y_D(a,\z) c \,.
\end{align*}
This proves Eq.\ \eqref{e-loc4}.
\end{proof}

Combining the results of Sections \ref{Se.3.1}--\ref{sse5.7n}
completes the proof of Theorem \ref{TH}.

\appendix
\renewcommand{\theequation}{\Alph{section}.\arabic{equation}}

\section[${}$\hspace{45pt}Representation Theory Point of View]{Representation Theory Point of View}\label{sse5.3n}

In this appendix, we present an alternative construction of $Y_D (a,\z)b$.
Fix an integer $n$ and two elements $a\in V_{\Delta'}$ and $b\in V_{\Delta''}$.
As in Sect.\ \ref{Se.3.1}, due to the integrability assumption from Definition \ref{main-def}$(a)$, without loss of generality we can assume
that all eigenspaces $V_\Delta$ of $H$ are finite dimensional; since we can choose
finite-dimensional $\so(D,\C)$--representations $W'\subseteq V_{\Delta'}$ and $W''\subseteq V_{\Delta''}$
such that $a\in W'$ and $b\in W''$.

Then the map $\mu_n^1$, defined by Eq.\ (\ref{vert44}),  restricts to an element
$\mu_n^1$ $\in$ $\Hom(V_{\Delta'}$ $\otimes$ $V_{\Delta''},$ $V_{\Delta'+\Delta''+n})$
as in Eq.\ \eqref{mu1-graded}.
Recall from the proof of \leref{lm5.x1} that 
$\mu^1_n$ is $\so(D-1,\C)$--equivariant, where the $\so(D-1,\C)$--subalgebra of
$\so(D,\C)$ is spanned by
$\Om_{\alpha\beta}$ for $2 \leqslant\alpha, \beta\leqslant D$.
Thus, $\mu_n^1$ is an invariant 
element of $\Hom(V_{\Delta'} \otimes V_{\Delta''}, V_{\Delta'+\Delta''+n})$
under the subalgebra $\so(D-1,\C)$,
i.e., the action of $\so(D-1,\C)$ on it is trivial.
Hence, we can apply the following lemma.

\begin{lemma}\label{Lm.3.1}
Let\/ $F$ be a representation of\/ $\so(D,\C)$, which is decomposable into a direct sum of
finite-dimensional irreducible\/ $\so(D,\C)$--representations.
Assume that\/ $v \in F$ is invariant under the action of the subalgebra\/ $\so(D-1,\C)$.
Then\/ $v$ is contained in a subrepresentation of the type\/
$\mathop{\bigoplus}\limits_{m \, = \, 0}^{\infty} Q_m \otimes \C_m^{\harm} [\z]$,
where\/ $\C_m^{\harm} [\z]$ are the\/ $\so(D,\C)$--representations of
degree\/ $m$ homogeneous harmonic polynomials, and\/ $Q_m$ are multiplicity spaces with only a finite number of them non-zero.
\end{lemma}

\begin{proof}
Recall that the Lie algebra $\so(D,\C)$ is of type $D_l$ for $D=2l$ and of type $B_l$ for $D=2l+1$. In either case it has rank $l$.
Denote by $R(\Lambda)$ the irreducible highest weight $\so(D,\C)$--module with  
highest weight $\Lambda$, and denote by $\pi_1,\dots,\pi_l$ the fundamental 
weights. Then $R(\pi_1) \cong \C^D$ is the vector representation.
Its symmetric powers are (see e.g.\ \cite{GW09,OV90}):
\beq\label{so1}
S^m \C^D \cong \bigoplus_{0 \leqslant k \leqslant \frac{m}2} R((m-2k)\pi_1) \,.
\eeq
Therefore, an irreducible $\so(D,\C)$--module $R(\Lambda)$ is contained
in $S^* \C^D \cong \C[\z]$ if and only if $\Lambda=m\pi_1$ for some $m\in\ZP$.
It is well known that the space $\C_m^{\harm}[\z]$ is an irreducible
$\so(D,\C)$--representation with highest weight $m\pi_1$, i.e.,
\beq\label{so-m}
\C_m^{\harm}[\z] \cong R(m\pi_1) \,.
\eeq
In fact, the decomposition \eqref{so1} corresponds to the expression \eqref{harexp} of a degree $m$ homogeneous polynomial in terms of $\z^2$
and harmonic polynomials.
The proof now follows immediately from the next lemma.
\end{proof}

\begin{lemma}\label{lm-so1}
Let\/ $R$ be an irreducible\/ $\so(D,\C)$--module. Assume that there exists
a non-zero vector\/ $v\in R$ annihilated by the subalgebra\/ $\so(D-1,\C)$.
Then\/ $R\cong \C_m^{\harm}[\z]$ for some\/ $m\in\ZP$.
Moreover, the vector\/ $v$ is unique up to a scalar multiple.
\end{lemma}
\begin{proof}
For $D=2l$, in accordance with the branching rule (see e.g. \cite[Theorem 8.1.4]{GW09}),
the restriction of $R(\lambda)$
to the subalgebra $\so(2l-1,\C) $ is given by
\beq\label{so4}
R(\lambda)\big|^{}_{\so(2l-1,\C) } \cong \bigoplus_\nu R'(\nu) \,.
\end{equation}
Here $R'(\nu)$ is the irreducible finite-dimensional representation of
$\so(2l-1,\C) $ with highest weight $\nu$, and the sum is taken over
all weights $\nu$ satisfying the inequalities
\beq\label{so5}
\lambda_1\geqslant\nu_1\geqslant\lambda_2\geqslant\nu_2\geqslant\cdots\geqslant\lambda_{l-1} \geqslant
\nu_{l-1}\geqslant|\lambda_{l}| \,,
\end{equation}
with all the $\nu_i$ being simultaneously
integers or half-integers together with the $\lambda_i$.

For completeness, let us also consider the case $D=2l+1$.
Then, in accordance with the branching rule (see e.g. \cite[Theorem 8.1.3]{GW09}),
the restriction of $R(\lambda)$
to the subalgebra $\so(2l,\C) $ is given by
\beq\label{so2}
R(\lambda)\big|^{}_{\so(2l,\C) } \cong \bigoplus_{\nu} R'(\nu) \,,
\end{equation}
where $R'(\nu)$ is the irreducible finite-dimensional representation of
$\so(2l,\C) $ with highest weight $\nu$, and the sum is taken over
all weights $\nu$ satisfying the inequalities
\beq\label{so3}
\lambda_1\geqslant\nu_1\geqslant\lambda_2\geqslant\nu_2\geqslant\cdots\geqslant
\lambda_{l-1}\geqslant \nu_{l-1}\geqslant\lambda_{l}\geqslant|\nu_l| \,,
\end{equation}
with all the $\nu_i$ being simultaneously
integers or half-integers together with the $\lambda_i$.

Now let us assume that the restriction $R(\lambda)\big|^{}_{\so(D-1,\C)}$
contains the trivial representation $R'(0)$.
Then inequalities \eqref{so5} and \eqref{so3} imply that
for $\nu=(0,\dots,0)$, one has $\lambda=(\lambda_1,0,\dots,0)$, where
$\lambda_1$ is a non-negative integer. Therefore $\lambda=m\pi_1$
for some $m \in\ZP$. This means that $R\cong R(m\pi_1) \cong \C_m^{\harm}[\z]$.

Finally, the uniqueness of $v$ follows from the fact that the decompositions \eqref{so4} and \eqref{so2} are multiplicity free; in particular, 
the trivial $\so(D-1,\C)$--module $R'(0)$ appears only once in them.
 \end{proof}

We remark that the unique (up to a scalar multiple) $\so(D-1,\C)$--invariant element $h_m(\z)\in\C_m^{\harm}[\z]$ can be obtained as the image
of the projection of $(z^1)^m \in S^m\C^D$ onto the summand $R(m\pi_1)\cong \C_m^{\harm}[\z]$ in the decomposition \eqref{so1}.
This implies that $h_m(\z)$ can be normalized so that
$
h_m(x\e_1) = x^m 
$.

\begin{remark}
A more explicit expression for $h_m(\z)$ can be derived from 
\cite[Eqs.\ (3.25), (3.29)]{BN06}.
In more details, a generating series for $h_m(\z)$ is the expansion
\begin{align*}
\iota_{\e_1,\z} \bigl((\e_1+\z)^2\bigr)^{-(D-2)/2}
&:= e^{\z \spr \di_{\w}} (\w^2)^{-(D-2)/2} \Bigr|_{\w=\e_1} \\
&=
\mathop{\sum}\limits_{m=0}^{\infty} 
\binom{-D+2}{m}
h_m(\z)
\,.
\end{align*}
We can further write
\[
h_m(\z) = \frac{(\z^2)^{m/2}}{\binom{-D+2}{m}} \, C^{((D-2)/2)}_m \bigl( -z^1 (\z^2)^{-1/2} \bigr)
\]
in terms of the Gegenbauer polynomials $C^{(\alpha)}_m(x)$, which are defined by the expansion
\[
(1-2xt+t^2)^{-\alpha} = \mathop{\sum}\limits_{m=0}^{\infty} C^{(\alpha)}_m(x) \, t^m
\,, \qquad 0\leqslant |x| <1 \,, \; |t| \leqslant 1 \,, \; \alpha>0 \,.
\]
\end{remark}

Let us fix a basis $\{h_{m,\sigma}(\z)\}_{\sigma=1}^{\mathfrak{h}_m^D}$ for $\C_m^{\harm}[\z]$ such that $h_{m,1}(\z)=h_m(\z)$ 
and
\beq\label{so6}
h_{m,\sigma}(x\e_1) = \delta_{\sigma,1} x^m \,;
\eeq
the last property can be achieved by subtracting from each $h_{m,\sigma}$ a suitable multiple of $h_{m,1}$.
As a consequence of Lemmas \ref{Lm.3.1} and \ref{lm-so1}, we obtain that in the space $\Hom(V_{\Delta'} \otimes V_{\Delta''},V_{\Delta'+\Delta''+n})$
there is a system of linearly independent elements $\{f^n_{k,m,\sigma}\}_{k,m,\sigma}$ such that:
\begin{enumerate}
\medskip
\item[$(a)$]
For every fixed $n$, $k$ and $m$, the subsystem $\{f^n_{k,m,\sigma}\}_{\sigma=1}^{\mathfrak{h}_m^D}$ is a basis of an irreducible $\so(D,\C)$--subrepresentation isomorphic to $\C_m^{\harm}[\z]$,
corresponding to the fixed basis $\{h_{m,\sigma }(\z)\}_{\sigma=1}^{\mathfrak{h}_m^D}$. 

\medskip
\item[$(b)$]
There is a decomposition
\(
\mu_n^1\big|_{V_{\Delta'} \otimes V_{\Delta''}}
= \displaystyle\sum\limits_{k,m} \gamma^{n}_{k,m} f^n_{k,m,1}
\)
where the $\gamma$'s are complex numbers.
\end{enumerate}
\medskip
Then we can define
\beq\label{HSFC}
Y_D(a,\z)b \, := \,
\mathop{\sum}\limits_{n,k,m,\sigma} \gamma^{n}_{k,m} \,
f^n_{k,m,\sigma} (a \otimes b) \,
\bigl(\z^2\bigr)^{\frac{n-m}{2}} \, h_{m,\sigma} (\z)
\,.
\eeq
By construction, this expression is $\SO(D,\C)$--equivariant; hence, it does not depend on the choice of basis $\{h_{m,\sigma }(\z)\}_{\sigma=1}^{\mathfrak{h}_m^D}$.
Moreover, $Y_D(a,x\e_1)=Y_1(a,x)$ by the assumption \eqref{so6}.
Thus, \eqref{HSFC} agrees with our previous definition, due to the uniqueness from \leref{SOD-covar}.

\section*{Declarations}

\subsection*{Funding}
B.N.B. was supported in part by a Simons Foundation grant 584741. 
Both authors were supported in part by the Bulgarian National Science Fund under research grant KP-06-N68/3.

\subsection*{Competing interests}
The authors have no competing interests to declare that are relevant to the content of this article.

\subsection*{Data availability statement}
Data sharing not applicable to this article as no datasets were generated or analysed during the current study.

\subsection*{Acknowledgements}
B.N.B. is grateful to the Institute for Nuclear Research and Nuclear Energy of the Bulgarian Academy of Sciences for the hospitality 
during the final stages of this work in June 2022. 
We would like to thank Ludmil Hadjiivanov and the referees for their comments, which helped improve the exposition of the paper.

\end{document}